\RequirePackage{fix-cm}
\documentclass{svjour3}                     % onecolumn (standard format)

\usepackage[T1]{fontenc}
\usepackage[
backend=biber,
style=alphabetic,
sorting=ynt
]{biblatex}

\usepackage{wrapfig}

\usepackage{graphicx}
\usepackage{amsthm}
\usepackage{musicography}
\usepackage{enumerate}

\usepackage{thmtools, thm-restate}
\usepackage{algorithm, algpseudocode}
\newcommand {\R}{\mathbb{R}}

\newcommand{\calD}{\mathcal{D}}

\newcommand{\calE}{\mathcal{E}}

\newcommand{\calS}{\mathcal{S}}

\DeclareMathAlphabet{\mathdutchcal}{U}{dutchcal}{m}{n}

\usepackage{caption}
\usepackage{float}
\usepackage{subcaption}
\usepackage{amsfonts}
\usepackage{amsmath}
\usepackage{amssymb}
\usepackage{amsbsy} 
\usepackage{mathrsfs}
\usepackage{multicol}
\usepackage{mathrsfs} 
\usepackage[all,cmtip]{xy}

\usepackage{xpatch}
\usepackage{wrapfig}

\DeclareMathOperator{\rank}{rank}

\newcommand{\M}{\mathcal{M}}

\DeclareMathOperator{\Lie}{Lie}

\DeclareMathOperator*{\argmin}{argmin}

\numberwithin{equation}{section}

\newcommand*{\defeq}{\mathrel{\vcenter{\baselineskip0.5ex \lineskiplimit0pt
                     \hbox{\scriptsize.}\hbox{\scriptsize.}}}%
                     =}

\begin{document}

\title{Principal subbundles for dimension reduction}

\author{Morten Akhøj$^{1,2, \ast}$ \and James Benn$^{1}$ 
 \and Erlend Grong$^{3}$ \and
Stefan Sommer$^{2}$  \and Xavier Pennec$^{1}$}
\authorrunning{M. Akhøj et al.}

\institute{
$^{\ast}$ Corresponding author
\\
$^1$ Université Côte d'Azur and INRIA, Sophia Antipolis, France  \\
$^2$ Department of Computer Science, University of Copenhagen, Denmark
\\
$^3$ Department of mathematics, University of Bergen, Norway
\\
\email{\{morten.pedersen, james.benn, xavier.pennec\}@inria.fr, sommer@di.ku.dk, Erlend.Grong@uib.no} 
}

\date{Submitted July 6, 2023}

\maketitle

\begin{abstract} 
In this paper we demonstrate how sub-Riemannian geometry can be used for manifold learning and surface reconstruction by combining local linear approximations of a point cloud to obtain lower dimensional bundles.
Local approximations obtained by local PCAs are collected into a rank \textit{k} tangent subbundle on $\mathbb{R}^d$, $k<d$, which we call a principal subbundle. This determines a sub-Riemannian metric on $\R^d$. We show that sub-Riemannian geodesics with respect to this metric can successfully be applied to a number of important problems, such as: explicit construction of an approximating submanifold $M$, construction of a representation of the point-cloud in $\R^k$, and computation of distances between observations, taking the learned geometry into account. The reconstruction is guaranteed to equal the true submanifold in the limit case where tangent spaces are estimated exactly. Via simulations, we show that the framework is robust when applied to noisy data. Furthermore, the framework generalizes to observations on an a priori known Riemannian manifold.
    \end{abstract}

\section{Introduction}
This paper presents a framework for learning an unknown, lower dimensional geometry from a set of observations $\{x_1, \dots, x_N\}$ in $\R^d$, or more generally on a Riemannian manifold. In our presentation we will assume $\R^d$-valued data unless otherwise specified - the case of manifold-valued data is presented in Section \ref{sect_manif_valued_case}. The framework provides concrete methods for solving the following three problems, 

\begin{figure}[h]
\centering
\includegraphics[scale=0.3]{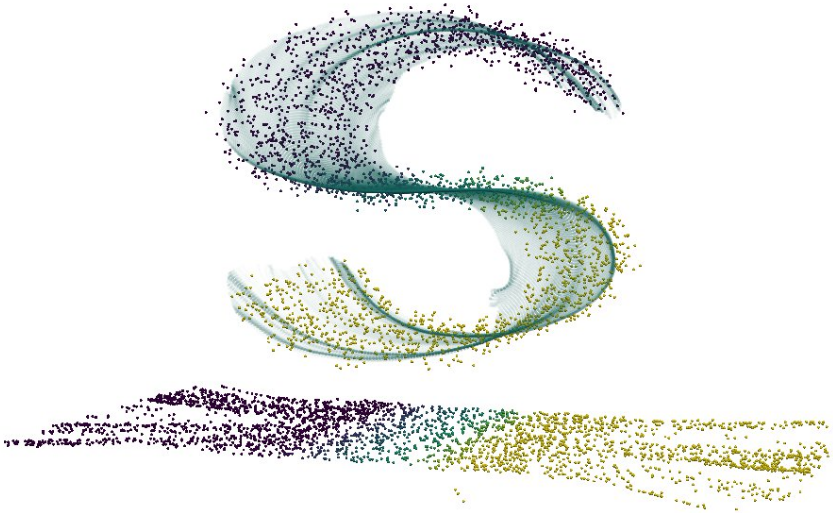}
\caption{\textit{Top}: noisy observations on the S-surface, embedded in $\R^{100}$ but projected to $\R^3$ for the purpose of visualization. The turquoise surface shows the 2-dimensional manifold reconstructed using the principal subbundle. \textit{Below}: a 2D tangent space representation of the observations. Experiment described further in Section \ref{sect_s_surface}.}
\label{fig_s_surface}
\end{figure}

\begin{enumerate}[\rm (A)]
    \item \textit{Metric learning}, i.e. learning a distance metric, $d(\cdot, \cdot) : \R^d \times \R^d \to \R_{\geq0}$, expressing the unknown underlying geometry (see \cite{bellet2015metric} for an overview). \vspace{3mm}
    \item \textit{Manifold reconstruction}, i.e. estimating a $k$-dimensional smooth submanifold $M\subset \R^d$ around which the data might be assumed to be distributed, an assumption known as the \textit{manifold hypothesis} \cite{cayton2005algorithms}. This includes surface reconstruction for observations in $\R^3$ \cite{surveySurfacReconstructionHuang}. \vspace{3mm}
    \item \textit{Dimension reduction}, in the specific sense of learning a representation of the data in $\R^k$, $k < d$, that preserves various chosen local properties, e.g. pairwise distances and angles between neighbouring points. This problem is often called \textit{manifold learning} \cite{ma2012manifold}, referring to the fact that the manifold hypothesis is often assumed, although most such methods do not reconstruct the manifold in $\R^d$.
\end{enumerate}

Each of these problems constitutes a whole field of research in itself. Indeed, their assumptions on the data can differ; while methods for (B) and (C) assume a lower dimensional structure of the data, this is not necessarily the case in (A). However, the framework described in this paper can be used to do both (A), (B) and (C). Our basic assumption is that the data is locally linear, i.e. locally well approximated by $k$-dimensional affine linear subspaces. This assumption holds under the manifold hypothesis, where the tangent space at each point is a good approximation. 
However, the assumption may also hold even if the manifold hypothesis fails, due to the phenomenon of \textit{non-integrability} (see Section \ref{sect_integrability}). In this sense, the framework of principal subbundles relaxes the manifold assumption.

At each point in $\R^d$ we estimate a \textit{k}-dimensional linear approximation of the data by an eigenspace of a local principal component analysis (PCA). Technically, the collection of these eigenspaces forms a \textit{subbundle} on $\R^d$. In this work we exploit the fact that such a subbundle determines a \textit{sub-Riemannian metric} on $\R^d$. Under such a metric a curve in $\R^d$ has finite length if and only if it is \textit{horizontal}, i.e. if its velocity vector lies within the subbundle at all time points. Due to the nature of the chosen subbundle, a horizontal curve initialized within the point cloud is expected to evolve along the point cloud. Thus, our framework provides a method for \textit{metric learning} (A) in the sense that it estimates a sub-Riemannian metric on $\R^d$, which, under certain assumptions, induces a distance metric on $\R^d$. In particular, it is a geodesic distance, meaning that the distance between $p,q\in\R^d$ equals the length of the shortest horizontal curve connecting $p$ and $q$. A sub-Riemannian metric can be thought of as a  Riemannian metric of lower rank $k\leq d$. To the best of our knowledge,  the low-rank (i.e. sub-Riemannian) case has not yet been explored in Riemannian approaches to metric learning (e.g. \cite{hauberg2012geometric}, \cite{perraul2013non}). But it is exactly this property that enables the metric to also provide solutions to problems (B) and (C). It yields a method for \textit{manifold reconstruction} (B) since the sub-Riemannian metric induces a diffeomorphism, $\phi_{\mu} : \R^k \supset U \to \phi_{\mu}(U) \subset \R^d$, whose image is a smooth \textit{k}-dimensional submanifold $M^k$ approximating the data around a chosen base point $\mu \in \R^d$. Technically, $\phi_{\mu}$ is a restriction of the sub-Riemannian exponential map at $\mu$. Finally, the framework yields a method for dimension reduction (C) since $U \subset \R^k$ is a coordinate chart for the manifold, so that, after projection of the observations to $M^k$, each projected observation $x_i$ can be represented as $\phi_{\mu}^{-1}(x_i) \in \R^k$.

Methods for manifold reconstruction (B) and dimension reduction (C)) often deal with the problem of how to combine local linear approximations into a global, non-linear representation. In the field of surface reconstruction from 3D point clouds, state-of-the-art methods such as \textit{Poisson surface reconstruction} (PSR)~\cite{kazhdan2006poisson} and \textit{Implicit Geometric Regularization} (IGR) \cite{gropp2020implicit}  are based on estimation of tangent spaces, which is done via estimation of normals (see \cite{surveySurfacReconstructionHuang} for a survey and benchmarking). A fundamental obstacle to this strategy of reconstructing a submanifold from tangent space approximations, e.g. reconstructing a surface from a normal field, is that the subspaces determine a submanifold if and only if they form an \textit{integrable} subbundle, cf. the Frobenius theorem (see Section \ref{sect_integrability} below). If the subspaces are estimated from a (finite) set of observations, integrability cannot be assumed to hold, even in the absence of noise. 
PSR and IGP deal with this problem by finding a surface whose normals minimize the distance to the empirical (noisy) normals. This surface is constructed by solving a Poisson equation (PSR) or by fitting a neural network (IGR). However, the approach of fitting normals does not generalize to the case of codimension greater than one, since normals are not defined in this case. Likewise, within manifold learning, methods based on alignments of local linear approximations (e.g. \cite{teh2002automatic}, \cite{zhang2004principal}, \cite{singerVectorDiffusionMaps}, \cite{koelle2022manifold}, \cite{myhre2020generic}), can be thought of as different ways to deal with non-integrability. Such methods are often based on eigendecomposition of a kernel-type matrix, or other linear-algebraic computations. This strategy is useful for finding a representation in $\R^k$ (problem (C)) but not for reconstructing an underlying manifold (problem (B)). The approach presented in this paper is different. We combine the local linear approximations into a global representation by integrating a system of second-order ordinary differential equations, the  sub-Riemannian geodesic equations. For $k=1$, this integration yields the flow of the first eigenvector field, called the \textit{principal flow} in \cite{panaretos2014principal}. There are, however, important differences between principal flows and our framework for $k=1$, see the discussion in Section \ref{sect_rel_to_principal_flows} and numerical results in Section \ref{sect_applications_manif_data}. A follow-up work to the principal flows is the \textit{Principal submanifolds} \cite{princSubmYaoEltzner}, where the aim is to leverage $k$ eigenvectors to construct a \textit{k}-dimensional submanifold approximating the data. This method is closely related to ours, in that it is based on horizontal curves. A crucial difference, however, is that the curves in \cite{princSubmYaoEltzner} are defined by an algorithmic procedure with no theoretical guarantees and the output of the method is a subset of the ambient space whose properties are largely unknown, such as whether it is in fact a submanifold.

A basic motivation and justification for our method is the following observation: if one had access to the true tangent spaces, e.g. via a frame of vector fields spanning them, then the Riemannian geodesic equation w.r.t. the corresponding Riemannian metric will generate an open subset (a normal chart) of the true manifold. I.e. it will generate an exact reconstruction, locally. When the frame is non-integrable, which is likely the case when it is estimated from data,  the more general sub-Riemannian framework is needed. We show that, surprisingly, we can still generate a submanifold in this setting, and thereby give solutions to problems (B) and (C). Our framework thus offers a new way to form a global representation from local linear ones that seems natural from the point of view of differential geometry.

\vspace{3mm}

\noindent \textbf{Contributions and overview of the paper}

Our main contribution is the idea of collecting local PCA's subspaces into a tangent subbundle and showing how the induced sub-Riemannian structure can be used to model the data. In Section \ref{sect_principal_subbundles}, we define principal subbundles on $\R^d$ and prove smoothness properties. In Section \ref{sect_sr_geom_eucl} we present sub-Riemannian geometry on $\R^d$. A large part of this section is devoted to background theory, with some exceptions, e.g. subsection \ref{sect_exp_image_dual_subbundle} where we prove that a certain restriction of the sub-Riemannian exponential map is a diffeomorphism, thus generating a submanifold even if the subbundle is non-integrable. This is the crucial result showing the usefulness of sub-Riemannian geometry for manifold reconstruction. In Section \ref{sect_sr_geometry_of_principal_subbundle} we discuss the particular sub-Riemannian geometry induced by the principal subbundle. In Section~\ref{sect_manif_valued_case} we show that the framework generalizes to the case of observations on an a priori known Riemannian manifold. Section \ref{sect_applications} presents numerical solutions to examples of problems (A) (metric learning), (B) (manifold reconstruction) and (C) (dimension reduction) for observations in $\R^d$ and on the sphere.

%[discarded] The Framework can be summarized as: we perform local PCA's at, in theory, every point $p\in\R^d$, extracting the first \textit{k} eigenvectors of maximal variance. We collect these into a rank \textit{k} subbundle on $\R^d$, and consider the sub-Riemannian structure resulting by restricting the Euclidean metric to this subbundle. Sub-Riemannian geodesics w.r.t. this structure are (locally) as short as possible, while satisfying the constraint of moving tangentially to $\mathcal{E}^{k}_p$ - and we can define a length on $\R^d$ as the length of the shortest geodesic between two points. The SR exponential map, from a chosen base point $\mu$, can now be used to construct a $k$-dimensional submanifold approximating the data around $\mu$. The method does not involve learning of parameters. It depends on a few hyperparameters which are common to other methods; a kernel range, a pre-specified dimension $k<d$, and (in the case of the manifold construction) a base point $\mu$.  We stress that constructing an approximating submanifold is not the only use of a principal subbundle - the SR geodesics and the induced geodesic distance is a more general concept. 

\section{Principal subbundles}\label{sect_principal_subbundles}

In this section, we define the \textit{principal subbundle} as a collection of eigenspaces of local PCAs. Recall that the tangent bundle on $\R^d$, $T\R^d$, can be identified with $\R^d \times \R^d$. For some subset $U \subset \R^d$, the tangent bundle on $U$, $TU \subset T\R^d$, can be identified with $U \times \R^d$. A rank \textit{k} subbundle $\calD$ of $TU$ is a collection of \textit{k}-dimensional subspaces associated to points in $U$,  that is $$\calD = \{(x,v) \hspace{0.5mm}\vert\hspace{0.5mm} x\in U, v\in\calD_x \},$$ where each $\calD_x$ is a \textit{k}-dimensional subspace of $\R^d$. Given a data set $\{x_i\}_{i=1..N} \subset \R^d$, we will define the principal subbundle as the subbundle for which each $\calD_x$ is the span of the first \textit{k} eigenvectors of a centered local PCA computed at $x \in \R^d$. We detail this construction below.

\subsection{Local PCA at the local mean}

Let $x_1, \dots, x_N$ be observations in $\mathbb{R}^d$. By \textit{local PCA} at $p\in\R^d$ we mean the extraction of eigenvectors of the following weighted and centered second moment.

\begin{definition}[Weighted, centered first and second moments]\label{def_weighted_2nd_moment} 
Let $K_{\alpha} : \mathbb{R}_{\geq 0} \to \mathbb{R}_{\geq 0}$ be a smooth, decaying 
 kernel function with range parameter $\alpha > 0$. 
At a point $p \in \mathbb{R}^d$, the normalized weight of observation $x_i$ is 
\[
w_i(p) := \frac{ K_\alpha(\Vert x_i - p \Vert)}{\sum_{j=1}^N K_\alpha(\Vert x_i - p \Vert)} ,\]
where $\Vert \cdot \Vert$ is the standard norm on $\R^d$. The weighted first moment (the local mean) and the centered weighted second moment (the local covariance matrix) are then:
\begin{align*}
m(p) = \sum_{i=1}^N w_i(p) \: x_i \quad , \quad 
\Sigma_{\alpha}(p) := \sum_{i=1}^{N} w_i(m(p)) (x_i - m(p)) (x_i - m(p))^T \in \mathbb{R}^{d \times d}.
\end{align*}
\end{definition}
\begin{remark}
    To save computational time, instead of using $w_i(m(p))$ in $\Sigma_\alpha(p)$ we suggest to use $w_i(p)$, i.e. not recomputing the weights at $m(p)$. This cheaper version is used for the experiments in Sections \ref{sect_surface_recon}-\ref{sect_learning_distance_metric}.
\end{remark}

For $K_{\alpha}$ constantly equal to 1 (or $\alpha = \infty$), $\Sigma_{\alpha}(p)$ is the ordinary mean-centered covariance matrix, independent of \textit{p}. In our experiments we use a gaussian kernel with standard deviation $\alpha$. A motivation for using local PCA's is the following. Under the manifold hypothesis, with an underlying manifold of dimension $k$, the \textit{k}-dimensional eigenspace of a local PCA at an observation $x_i$ converges to the true tangent space of that submanifold at $x_i$ in the limit of zero noise and the number of observations going to infinity (see e.g. \cite{singerVectorDiffusionMaps}, Theorem B.1, for a convergence result).

\subsection{Eigenvector fields and the principal subbundle}\label{section_eigenbundle}

We define the principal subbundle at $p \in \R^d$ as a \textit{k}-dimensional eigenspace of the weighted second moment at \textit{p}. For it to be well-defined at $p$, the \textit{k}'th and \textit{k+1}'th eigenvalues of the second moment at \textit{p} should be different. I.e. the subbundle is defined only outside the following set of points, which we will call singular,
\begin{equation}\label{eq_singular_set}
\mathcal{S}_{\alpha, k} \defeq \left\{p \in \R^d \hspace{1mm} \middle \vert  \hspace{1mm}  \lambda_{k}(p) = \lambda_{k+1}(p) \right\}, \quad 1 \leq k \leq d
\end{equation}
where $\lambda_1(p) \geq \dots \geq \lambda_d(p)$ are the eigenvalues of $\Sigma_{\alpha}(p) \in \R^{d\times d}$.

\begin{definition}[Principal subbundle]\label{def_PS} Let $\lambda_1(p) \geq \dots \geq \lambda_d(p)$ be the eigenvalues of $\Sigma_{\alpha}(p) \in \R^{d\times d}$
with associated eigenvectors $e_1(p),\dots, e_d(p)$. 
Let $\mathcal{S}_{\alpha, k}$ be the set of singular points (Eq. (\ref{eq_singular_set})). 
Then the \emph{principal subbundle} on $\R^d \setminus \mathcal{S}_{\alpha, k}$ is defined as $$\mathcal{E}^{k, \alpha} = \left\{(p,v) \hspace{1mm}\vert\hspace{1mm} p\in \R^d \setminus \mathcal{S}_{\alpha, k}, v\in \text{span}\{e_1(p),\dots, e_k(p)\}\right\} \subset T (\R^d \setminus \mathcal{S}_{\alpha, k}).$$
\end{definition}

\begin{remark}
    We consider it an assumption on the data, and the chosen parameters, that $\lambda_k(p) \neq \lambda_{k+1}(p)$ at all points where we want to evaluate the principal subbundle. In our  computations we have not encountered points where the assumption was violated. 
\end{remark}
\begin{remark}
    Cf. the proof of Proposition \ref{prop_smooth_princ_subb} (below), if $\lambda_k(p) > \lambda_{k+1}(p)$ at some $p \in \R^d$, then this property holds on an open set around $p$.
\end{remark}

Note that the principal subbundle only depends on the \emph{eigenspaces}, not the choice of eigenvectors. The latter are not uniquely determined, they depend on a choice of sign and, in the case of repeated eigenvalues, a rotation within a subspace. In order to define a sub-Riemannian structure from this subbundle it needs to be smooth, which is satisfied cf.  Proposition \ref{prop_smooth_princ_subb} below.  A closely related result, Lemma \ref{lem_eigenvectorfields} below, states that if an eigenvalue $\lambda'$ at $p\in \R^d$ has multiplicity 1, then there exists a smooth vector field on an open subset $O \subset \R^d$ around $p$ which is an eigenvector for $\Sigma_{\alpha}(x)$ at each $x \in O$. We call this vector field an \textit{eigenvector field}.

\begin{lemma}[Existence of smooth eigenvector fields]\label{lem_eigenvectorfields} Let $e'$ be an eigenvector of $\Sigma_{\alpha}(p)$ at $p \in \R^d$ with eigenvalue $\lambda'$ of multiplicity 1. Then there exists an open subset $O(p) \subset \R^d$ around \textit{p} and smooth maps $e : O(p) \to \R^{d}$ and $\lambda : O(p) \to \R_{\geq 0}$ satisfying $e(p) = e', \lambda(p) = \lambda'$, $\Vert e(x) \Vert = 1$ and $\Sigma_{\alpha}(x) e(x)= \lambda(x) e(x)$ for all $x \in O(p)$. \end{lemma}

This result follows directly from \cite{sun1985eigenvalues}, Theorem 2.3, since $\Sigma_{\alpha}$ is a smooth map. From this result on eigen\textit{vectors}, one can conclude that the eigen\textit{spaces} are smooth at $p$ if either the eigenvalues $\lambda_1(p), \dots, \lambda_{k+1}(p)$ are distinct, or $\lambda_{k}(p), \dots, \lambda_{d}(p)$ are distinct at $p \in \R^d$. However, we can in fact show smoothness of the subbundle under the milder, indeed minimal, condition that $\lambda_{k}(p) > \lambda_{k+1}(p)$ (Proposition \ref{prop_smooth_princ_subb}). Appendix \ref{app_proofs} contains the proof of this and all other results in the paper.  

 \begin{restatable}{proposition}{propSmoothPSrd}\label{prop_smooth_princ_subb}
 The principal subbundle, defined on $\mathbb{R}^d\setminus \mathcal{S}_{\alpha, k}$, is smooth.
\end{restatable}

Figure \ref{fig_ps_illustrations} illustrates the principal subbundle (blue arrows) induced by point clouds in $\R^2$ and $\R^3$, including the effect of centering the second moment at the local mean.

We are interested in studying curves whose velocity vectors are constrained to lie in the principal subbundle (i.e. eigenspaces of local PCA's). This can be done using sub-Riemannian geometry, which we introduce next.

\begin{figure}
  \centering
  \subfloat[a][]{\includegraphics[scale=0.37]{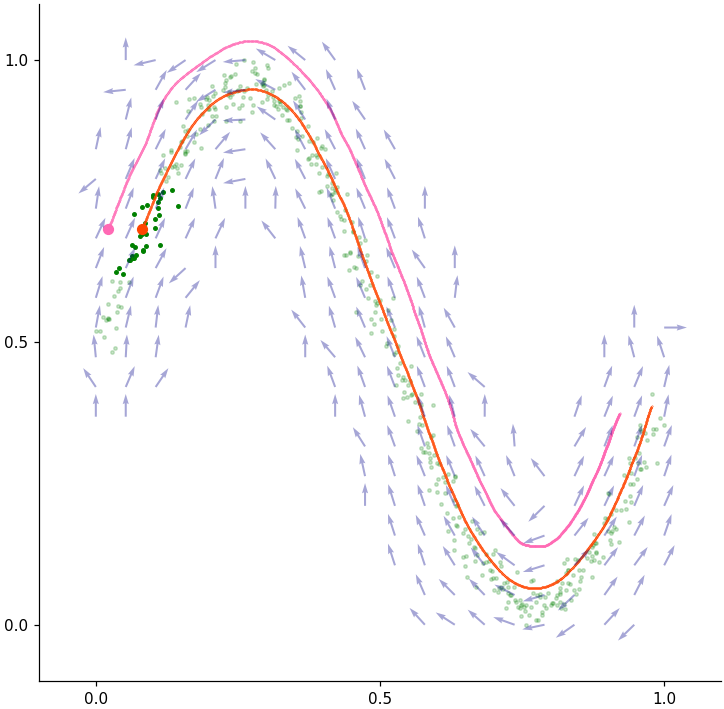}} \\
  \subfloat[b][]{\includegraphics[scale=0.35]{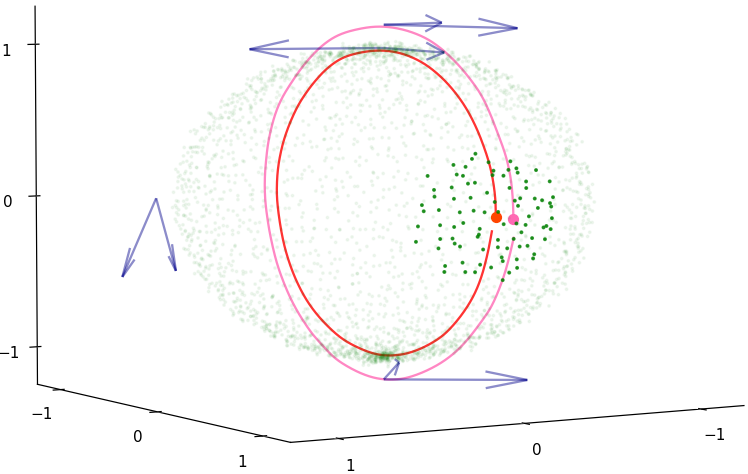}}
  \caption{Two illustrations of a principal subbundle induced by point clouds (green points) and sub-Riemannian geodesics (red and pink curves). On Figure \textit{(a)} the subbundle rank is $k=1$, on \textit{(b)} it is $k=2$. Blue arrows span the principal subbundle subspace at the basepoint of the arrows; on Figure \textit{a)} each subspace is a line (spanned by one arrow), on Figure \textit{(b)} they are planes (spanned by two arrows). The geodesics are initialized at the red, respectively pink, dots, which are inside, respectively outside, the point clouds. The observations colored in a 
 darker green are those with an assigned normalized weight $w_i$, w.r.t. the position of the red dot, larger than $10^{-5}$ - thus the lighter green observations has only a negligible effect on the local PCA computation. On Figure \textit{(b)}, the duration of integration, and thus the curve length, is $T=2\pi$ (red geodesic) and $T=2.3 \pi$ (pink geodesic), respectively - these are the circumferences of great circles centered at the origin and passing through the initial points.} \label{fig_ps_illustrations}
\end{figure}

\section{Sub-Riemannian geometry on $\R^d$}\label{sect_sr_geom_eucl}
%% In this section, we 
We now introduce basic notions of sub-Riemannian geometry on $\R^d$. We focus on the special case that we need, where the sub-Riemannian metric is a restriction of the standard Euclidean inner product. This viewpoint is not presented in sources that we know of, so we devote some space to it. For more comprehensive introductions see e.g. \cite{agrachev2019comprehensive} or \cite{jean2014control}. We strive to make the presentation accessible to someone with only a slight knowledge of differential geometry.

\subsection{Horizontal curves and the sub-Riemannian distance}

In the special case that we consider, a sub-Riemannian structure on $\R^d$ is fully determined by a rank $k$ subbundle $\mathcal{D} \subset T\R^d$. The subbundle can be represented as a smoothly varying orthogonal projection matrix, 
\begin{align}\label{eq_cometric_eucl}
    g^{\star} : \R^d \to \R^{d\times d} : p \mapsto F(p)F^T(p),
\end{align}
where $F : \R^d \to \R^{d\times k}$ is a smooth map s.t. $F(p)$ is a rank \textit{k} matrix whose columns form an orthonormal basis for $\calD_p$ at any $p\in\R^d$. The map $g^{\star}$ is called the \textit{cometric}. If $g^\star(p)$ has full rank $d$ at every $p\in\R^d$, then the map $p\mapsto g^\star(p)^{-1}$ is called a Riemannian metric. We discuss relations between Riemannian and sub-Riemannian geometries below.

A basic intuition behind sub-Riemannian geometry is that, at each point $p\in\R^d$, $\mathcal{D}_p$ contains the allowed velocity vectors of a curve passing through $p$. If a curve $\gamma : [0,1]\to\R^d$ satisfies  
$$\frac{d}{dt}\gamma(t) \defeq \dot{\gamma}(t) \in\mathcal{D}_{\gamma(t)}$$ for almost all $t \in [0,1]$ it is called \textit{horizontal}. This class of curves induces a \textit{distance metric} on $\R^d$, the \emph{Carnot-Carath\'eodory metric},
\begin{equation}\label{eq_carnot_carath_metric}
d^\mathcal{D}(p,q) = \inf \left\{ L(\gamma) \,\, \middle\vert
\begin{array}{c} \text{$\gamma:[0,1] \to \R^d$ is horizontal} \\ \gamma(0) =p, \gamma(1) = q\end{array}\right\} \in \R_{\geq 0} \cup \{\infty\} ,
\end{equation}
for any $p,q \in \R^d$, where $L(\gamma) \defeq \int_0^1 \Vert \dot \gamma(t) \Vert dt$ is the curve length functional. An important property of a sub-Riemannian geometry is whether any two points $p,q$ can be connected by a horizontal curve, or, equivalently, whether $d(p,q)$ is finite for all $p,q\in\R^d$. A sufficient condition for this is that $\mathcal{D}$ is \emph{bracket-generating} (cf. the Chow-Rashevski theorem, \cite{chow2002systeme}, \cite{agrachev2019comprehensive}). This means that, for all $p\in\R^d$, $\Lie \mathcal{D}_p$ equals $\R^d$, where $\Lie \mathcal{D}_p$ consists of the span of all $\mathcal{D}$-valued vector fields and all of their iterated Lie brackets (see e.g. \cite{lee2013smooth}). In this case, $d^\mathcal{D}$ induces the standard topology on $\R^d$.

\subsection{Sub-Riemannian geodesics}\label{sect_SR_geodesics}

%In this section, we discuss 
We now turn to horizontal curves that are 'locally length-minimizing', i.e. any local perturbation of the curve increases its length. For our purposes, the most important class of such curves is called \textit{normal sub-Riemannian geodesics}.

Normal geodesics are solutions to a system of equations on the cotangent bundle $T^{\star}\R^d$, which, in our setting, can be identified with $\R^d\times\R^d$. A curve $\gamma : [0, T] \to \R^d$ is a normal geodesic if and only if it is the projection to $\R^d$ of a curve in $T^{\star}\R^d$, $\psi : [0,1] \to T^{\star}\R^d$, that satisfies the \textit{sub-Riemannian Hamiltonian equations}.  Let \textit{H} denote the \textit{sub-Riemannian Hamiltonian}, $$H : T^{\star} \R^d \to \mathbb{R}_{\geq 0} : (p,\eta) \mapsto \frac{1}{2} \eta^T g_p^{\star} \eta.$$ We will write $H_p$ if we consider it as a function on $T^{\star}_p \R^d$ only. The Hamiltonian  equations are then given by 
\begin{align}\label{eq_SR_ham}
    \begin{split}
        \dot{p} &= \frac{\partial H}{\partial \eta}(p,\eta) = g_p^{\star} \eta, \\
        \dot{\eta} &= -\frac{\partial H}{\partial p}(p,\eta).
    \end{split}
\end{align}
A solution $\psi(t) \defeq (p(t), \eta(t))$ with initial value $(p_0, \eta_0)$ is called a \textit{normal extremal}. The associated normal geodesic is the curve $\gamma_{p_0}^{\eta_0}(t) \defeq \pi(p(t), \eta(t)) \defeq p(t)$, i.e. the projection of $\psi$ to the first component $\R^d$. Notice that the horizontality of $\gamma$ is apparent from the fact that $g^{\star}_p$ projects $\eta$ to $\calD_p$ in (\ref{eq_SR_ham}). In the Riemannian case the Hamiltonian equations are equivalent to a system of ODE's on the tangent bundle called the \textit{geodesic equations}. This parameterizes geodesics by their initial tangent vector instead of, as in the sub-Riemannian case, the initial cotangent vector. We end this section with a few facts about solutions to Hamilton's equations that we will need later on. Firstly, the Hamiltonian is conserved along solutions,
i.e. $H(p_t,\eta_t) = H(p_0,\eta_0)$ for all $t \in [0,T]$ (see e.g. \cite{agrachev2019comprehensive}, Section 4.2.1). This implies 
that a normal geodesic $\gamma$ is a constant speed curve, since 
\begin{align}\label{eq_constant_speed_ham}
    \Vert \dot{\gamma}(t) \Vert = \Vert g^\star_{p_t} 
\eta_t\Vert = \sqrt{2H(p_t,\eta_t)}.
\end{align} This further implies that 
$\gamma^{\eta_0}_{p_0}$ has unit speed if $\eta_0\in H_{p_0}^{-1}(1/2)$, and therefore that its length is given by the duration of integration $T$.
Lastly, we will need the fact that the Hamiltonian equations are time-homogenous in the sense that, for any $\eta_0 \in H^{-1}(1/2)$ and $\alpha > 0$, $\gamma^{\alpha\eta_0}_{p_0}(t) = \gamma^{\eta_0}_{p_0}(\alpha t)$ (\cite{agrachev2019comprehensive}, Section 8.6). 

\subsection{The sub-Riemannian $\exp$ and $\log$}\label{section_SR_exp}

The sub-Riemannian \textit{exponential} map at $p \in \mathbb{R}^d$ maps a cotangent $\eta \in T_p^{\star}\R^d \cong \R^d$ to the position at time 1 of the normal geodesic initialized by $(p, \eta)$, i.e. $$\exp^{\calD}_p \hspace{2mm} : \hspace{2mm} T^{\star}_p \R^d \to \R^d \hspace{2mm} : \hspace{2mm} \eta \mapsto \exp^{\calD}_p(\eta) \defeq \gamma_p^{\eta}(1).$$

\noindent The exponential map will also be denoted simply by $\exp$. The time-homogeneity of the Hamiltonian equations mentioned in the previous section has two important consequences. Firstly, for $\alpha>0$, it holds that $\exp_p(\alpha \eta) = \gamma_p^{\eta}(\alpha)$, so scaling $\eta$ amounts to moving along a single normal geodesic; secondly, $\gamma$ can be assumed to be unit speed parameterized, and therefore the length of the normal geodesic $\alpha \mapsto \exp_p(\alpha\eta)$, $\alpha \in [0,1]$, is given by $\sqrt{2H(\eta)}$. In the case where this normal geodesic is a global, not just local, length minimizer between its endpoints $p$ and $y \defeq \exp_p(\eta)$, we get the formula \begin{align}\label{eq_ham_vs_distance}
    d^\calD(p,y) = \sqrt{2H(p, \eta)}.
\end{align} 

\subsubsection{Optimizing for the $\log$}
To compute the sub-Riemannian distance between two points, eq. (\ref{eq_ham_vs_distance}) suggests that one should invert the exponential map. If the exponential map at \textit{p} is a diffeomorphism (thus invertible) around $0 \in T^{\star}_p\R^d$, its inverse is called the \textit{logarithmic} map, defined by $$\log^{\calD}_p : \R^d \supset U \to O \subset T^{\star}_p \R^d \quad \text{satisfying} \quad \gamma^{\log^{\calD}_p(y)}_p(1) = y$$ for some open sets $U$ and $O$ with $p\in U$. However, such an open set \textit{U} on which $\exp_p$ is a diffeomorphism only exists if $rank(\mathcal{D}) = d$ (see \cite{agrachev2019comprehensive} Prop. 8.40), in which case the geometry is Riemannian. A simple way to see this is that $\exp_p(H_p^{-1}(0)) = 0$, where $H_p^{-1}(0) = \calD_p^{\perp}$. In the sub-Riemannian case of $rank(\mathcal{D}) < d$ we propose an approximate log map given as a solution to the following optimization problem, for $p,y\in\R^d$,
\begin{align}\label{eq_log}    \overline{\log}_p(y) \in \underset{\eta \in \mathbb{A}}{\argmin} \hspace{1mm} \Vert \exp_p(\eta) - y \Vert^2 + H(p, \eta),
\end{align}
where $\mathbb{A} = T_p^{\star}\R^d$. This problem searches for the shortest normal geodesic between $p$ and $y$. For reasons that will be explained in Section \ref{sect_Rk_repr}, we will also be interested in the case of $\mathbb{A} = \calD_p^{\star} \subset T_p^{\star}\R^d$, the metric dual of $\calD_p$ (Equation \ref{eq_metric_dual} below). Under certain assumptions on $\calD$, notably bracket-generatingness, the image set $\exp_p(T^\star \R^d)$ is dense in $\R^d$ even when $\rank{\mathcal{D}} < d$ \cite{rifford2014sub}, implying that the error in (\ref{eq_log}) can be made arbitrarily small. The problem of finding shortest horizontal curves between points is studied in non-holonomic control theory (see e.g. \cite{jean2014control}). In our current implementations, however, we find (local) solutions via a minimization algorithm based on BFGS \cite{numerical_optimization_nocedal} and automatic differentation of the exponential map, which is possible using e.g. the python library Jax \cite{jax}.

\subsection{The subbundle induces a foliation}\label{sect_integrability}

If a bracket generating subbundle $\calD$ (i.e. $\Lie \calD = T\R^d$) represents one extreme for subbundles on $\R^d$ then its opposite is that of a constant rank integrable subbundle; that is, a constant rank subbundle $\tilde{\calD}$ satisfying $\Lie\tilde{\calD}= \tilde{\calD}$. An important property of integrable subbundles is that they posses integral manifolds which are immersed submanifolds $\M \subset \R^d$ such that $T_p \M = \tilde{\calD}_p$ for all points $p\in \M$. Given a constant rank integrable subbundle $\tilde{\calD}$, the global Frobenius Theorem tells us that $\R^d$ is foliated, or partitioned, by the collection of all maximal integral manifolds of $\tilde{\calD}$ - each integral manifold is called a leaf and has dimension equal to the rank of $\tilde{\calD}$ (see Lee, Chapter 19 for full details on integrable subbundles, there called \textit{involutive distributions}, and the Frobenius Theorem). The geometry induced by $\tilde{\calD}$ on $\M \subset \R^d$ is Riemannian since $\tilde{\calD}_p$ is the full tangent space at each point $p \in \M$, implying that all curves on $\M$ are horizontal; therefore the sub-Riemannian geodesic equations are identical to the Riemannian geodesic equations. If a subbundle 
$\breve{\calD}$ is neither bracket generating ($\Lie\breve{\calD}=T\R^d$) nor integrable ($\Lie \breve{\calD} = \breve{\calD}$) then the subbundle $\Lie\breve{\calD}\subset T\R^d$ is integrable and foliates $\R^d$ by its integral manifolds $\M$, each of dimension $\rank(\Lie \breve{\calD})$. The induced geometry on each integral manifold $\M$ is sub-Riemannian (not all curves are horizontal).

In relation to problem \textit{A}, mentioned in the introduction, the previous discussion implies that the induced distance metric is finite, $d^\calD(p,q) < \infty$, for all points $p,q$ in the same leaf, whereas it is infinite for points belonging to different leaves - a horizontal curve is constrained to move within a single leaf.
In relation to problem \textit{B}, we are interested in generating a
$k$-dimensional submanifold of $\R^d$ from a rank $k$ subbundle $\calD$ whose integrability
or bracket generation is a priori unknown. In Proposition 3.1 below we show how
this can be done via sub-Riemannian geometry. The generated submanifold is tangent to $\calD$ in 'radial' directions, but not in all directions, as will be explained
below.

\subsection{The exponential image of the dual subbundle}\label{sect_exp_image_dual_subbundle}

The content of the previous sections implies the following. If $\calD$ is integrable, then there exists an open set $U\subset \calD_p$ s.t. $M \defeq \exp_p(U)$ is a \textit{k}-dimensional embedded submanifold of $\R^d$ whose tangent space a every point $q\in M$ equals $\calD_q$. In this case, $\exp_p$ is a diffeomorphism from $U$ to this submanifold. On the other hand, if $\calD$ is not integrable, there exists no submanifold that is tangent to $\calD$, in particular $\exp_p(U)$ does not satisfy this. However, in the following we show that $\exp_p(U)$ is still a \textit{k}-dimensional embedded submanifold.

Let 
\begin{equation}\label{eq_metric_dual}
    \calD^\star_p \defeq \{\langle v, \cdot \rangle \vert v\in\calD_p\} \subset T^{\star}_p\R^d
\end{equation} 
be the dual space of $\calD_p$ w.r.t. the standard inner product $\langle v, u\rangle \defeq v^Tu$. This simply means that $\calD^\star_p$ consists of the tangent vectors (column vectors) in $\calD_p$ considered as covectors (row vectors). Thus, $\calD^\star_p$ is a \textit{k} dimensional subspace of $T^\star_p\R^d$ which can be identified with $\calD_p\subset T\R^d$.

\begin{proposition}[$\exp_\mu$ is a local diffeomorphism from $\calD^{\star}_\mu$]\label{prop_dual_subbundle_exp}\label{prop_sr_exp_diffeo}Let $\mu \in \R^d$ be arbitrary. 
There exists an open subset $C_\mu \subset \mathcal{D}_\mu^{\star}$ containing $0$ such that $\exp_\mu^{\mathcal{D}}$ restricted to $ C_\mu$ is a diffeomorphism onto its image. That is, $$M_\mu^{\mathcal{D}}
\defeq \exp_\mu^{\mathcal{D}}(C_\mu) \subset \R^d$$ is a smooth
$k$-dimensional embedded submanifold of $\mathbb{R}^d$ containing $\mu$.
\end{proposition}

It holds that $T_p(M^\calD_\mu) = \calD_p$ at $p=\mu$, but at a general $p \in M^\calD_\mu$ these spaces are different if $\calD$ is not integrable. They need not even be 'close', as can be seen in e.g. the Heisenberg group where $\exp_0^{\mathcal{D}}(C_0)$ is the xy-plane, to which the Heisenberg subbundle is almost orthogonal at certain points $p$. But $M^\calD_\mu$ is 'radially horizontal', in the sense that it is the union of normal geodesics from $\mu$ each of which is horizontal w.r.t. $\calD$. In particular, if we assume that $C_\mu$ is convex and let $\partial{C_\mu} \subset \calD^\star_p$ denote its boundary, then 
\begin{equation}\label{eq_princ_subm_union_geos}
    \exp_\mu^{\mathcal{D}}(C_\mu)  =  \left\{ \gamma_p^\eta(t) \hspace{1mm}\middle\vert \hspace{1mm} \eta \in \partial{C_\mu}, \hspace{1mm} t\in [0,1])\right\},
\end{equation}
\noindent where each geodesic $t \mapsto \gamma_p^\eta(t)$ is tangent to $\calD$.

Note that, since the exponential map restricted to $C_p$ is a diffeomorphism, the log-optimization problem (\ref{eq_log}) with $\mathbb{A} = \calD_p^{\star}$ has a unique solution for $p = \mu$ and any $y \in M^\calD_\mu$.

\section{Sub-Riemannian geometry of the principal subbundle}\label{sect_sr_geometry_of_principal_subbundle}

In this section, we present a sub-Riemannian (SR) structure on $\R^d$ based on local PCA's, namely, the SR structure determined by the principal subbundle. 
Moving horizontally with respect to the principal subbundle means to move within a \textit{k}-dimensional subspace of maximum local variation at each 
step. Therefore, geodesics that are horizontal w.r.t. this structure follow the point cloud, and the associated \textit{exp} and \textit{log} maps can be used for representing the 
data. The image of the dual subbundle under the exponential map, described in Proposition 
\ref{prop_dual_subbundle_exp} above, will be called a \textit{principal submanifold} when the principal subbundle is used. Such a submanifold approximates the data for well-chosen hyperparameters. This is described in Section \ref{sect_PS} where we also give an algorithm to compute it. Furthermore, we discuss the use of the log optimization problem (\ref{eq_log}) for giving a representation of the observations in $\R^k$ (Section \ref{sect_Rk_repr}) and for computing distances between observations (Section \ref{sect_sr_dist}).
 
\subsection{Properties of the sub-Riemannian structure}

The sub-Riemannian structure that we will use to model the data is the one determined by the principal subbundle $\calE^{k, \alpha}$, also denoted simply by $\calE$. 
%% Notation $\calE^{k, \alpha}$ was apparently not introduced before. 
Proposition \ref{prop_smooth_princ_subb} about smoothness of the subbundle implies smoothness of the cometric $g^{\star}$. For any $p\in\R^d \setminus \mathcal{S}_{\alpha,k}$ the cometric can be represented as $g_p^{\star} = F(p)F(p)^T \in \R^{d\times d}$, where $F = [e_1(p), \dots, e_k(p)]$ is a matrix whose columns are the first \textit{k} eigenvectors of the weighted second moment $\Sigma_{\alpha}(p)$ (Definition \ref{def_weighted_2nd_moment}).

We know that $\mathcal{E}$ is of constant rank \textit{k}, but we do not know  if $\Lie \mathcal{E}$ is of constant rank, let alone if it is bracket-generating (i.e. $\rank(\Lie \mathcal{E}) = d$). %However, $\Lie \mathcal{E}$ will be of constant rank locally around every point $p \in O$, where $O$ is an open and dense subset of $\R^d$ \cite{jean2014control}[,,,find ref eller slet].
%\todo{,,,? Just assume so? expand on the consequences of this.} 
Under the manifold hypothesis, in the limit of zero noise and the number of observations going to infinity, the convergence result of \cite{singerVectorDiffusionMaps} (Theorem B.1) suggests that the subbundle is everywhere tangent to a submanifold and thus integrable. 

\subsection{Computing geodesics}\label{sect_computing_geos_eucl}

We compute geodesics w.r.t. the chosen sub-Riemannian structure by numerically integrating the sub-Riemannian Hamiltonian equations (\ref{eq_SR_ham}), see Appendix \ref{app_implementation} for notes on the implementation. In \cite{sun1985eigenvalues}, Theorem 2.4, formulas are given for the derivatives of eigenvector fields. This enables computation of derivatives of the Hamiltonian, 
\begin{align*}
H(p,\eta) &= \frac{1}{2} \eta^T g_p^{\star} \eta \\
&= \frac{1}{2} \eta^T F(p)F(p)^T \eta \\
&= \frac{1}{2} \eta^T  [e_1(p), \dots, e_k(p)]  [e_1(p), \dots, e_k(p)]^T \eta,    
\end{align*}
via automatic differentiation libraries such as Jax \cite{jax}. The formulas in \cite{sun1985eigenvalues} hold under the assumption that the first $k+1$ eigenvalues, $\lambda_1(p), \dots, \lambda_{k+1}(p)$, are distinct (cf. Lemma \ref{lem_eigenvectorfields}). Note that our basic assumption on the observations is that they are well approximated locally by a $k$-dimensional linear space, implying that the first $k$ eigenvalues are relatively close, possibly equal. Two comments on this: \textit{1.} Using the results in \cite{sun1990multiple} (see also Proposition \ref{prop_smooth_princ_subb} and its proof), it is possible to compute derivatives of the Hamiltonian under the milder assumption of only $\lambda_k(p)$ and $\lambda_{k+1}(p)$ being distinct - however, in practice we have not had the need to pursue this. \textit{2.} Since the differences between $\lambda_1,\dots,\lambda_k$ are likely to be relatively small, the ordering and rotation of the eigenvectors is effectively random. However, this does not affect the Hamiltonian equations, since the Hamiltonian depends only on the cometric, a projection matrix, which is invariant to rotations and permutations of the basis $F(p)$ within $\calE_p$.
 
%At each integration step, the Hamiltonian needs to be evaluated (and differentiated), so it is necessary that none of the points along the trajectory belong to the singular set $\mathcal{S}_{\alpha, k}$. 
Figure \ref{fig_ps_illustrations} illustrates sub-Riemannian geodesics with respect to the metric induced by two different point clouds. The surfaces (principal submanifolds) presented in figures \ref{fig_s_surface} and \ref{fig_surface_recon} are likewise composed of many such geodesics, cf. the next section. 

\subsection{Principal submanifolds (Problem \textit{B})}\label{sect_PS}

As the first use of principal subbundles, we define the \textit{principal submanifold} from a base point $\mu \in \R^d \setminus \mathcal{S}_{\alpha, k}$, given a set of observations in $\R^d$. This choice of data representation implicitly assumes that the data is locally well-described by a submanifold, i.e. the 'manifold hypothesis'.

\begin{definition}[Principal submanifold at $\mu$] Let  $\{x_1, \dots, x_N\} \subset \R^d$ be a set of observations. Let $\mu\in \R^d \setminus \mathcal{S}_{\alpha, k}$ be a chosen base point, let $\alpha$ be the kernel range and let $k \in \{1, \dots, d-1\}$ be the rank of the principal subbundle, $\mathcal{E} = \mathcal{E}^{k, \alpha} \subset T\R^d$. Let $\mathcal{E}^\star_\mu$ be the dual subbundle at $\mu$, and $B_r \subset \mathcal{E}^\star_\mu$ a \emph{k}-dimensional open ball of radius \emph{r} containing $0$. The \emph{principal submanifold} of radius \emph{r} is given by
\begin{align}
    M^{k}_{\mu}(r) \defeq \exp^{\mathcal{E}}_{\mu}(B_r) \subset \R^d,
\end{align}
\end{definition}
\begin{remark} We will assume that \emph{r} is sufficiently small for $M^{k}_{\mu}(r)$ to actually be a submanifold, cf. Proposition \ref{prop_dual_subbundle_exp}. 
If we write simply $M^{k}_{\mu}$, we will assume that \emph{r} takes the largest such value.
\end{remark}

Algorithm \ref{algo_principal_submanifold} describes how to compute a point set 
representation of a principal submanifold, up to arbitrary resolution.  For hyperparameters, $\mu$, $k$, $\alpha$  
(the base point, dimension and range, respectively),
the principal submanifold, $M^{k}_{\mu}$, is an estimate of the true underlying submanifold, $M$, locally around $\mu$. As described in Section \ref{section_SR_exp}, $M^{k}_{\mu}$  cannot be expected to be exactly tangent to $\calE$ since $\calE$ might not be integrable. However, since $\mathcal{E}^{k, \alpha}$ approximates the tangent spaces of the true submanifold our expectation is that the subbundle is 'close' to being integrable and therefore that the difference between $\mathcal{E}_p$ and $T_{p}(M^k_{\mu})$ is small for $p\in M^k_{\mu}$. The approximation $M^k_{\mu} \approx M$ comes with the following guarantee: if $\mu \in M$ and the principal subbundle contains the true tangent spaces to $M$ around $\mu$, then the principal submanifold  is an open subset of the true submanifold $M$. In particular, the ball $B_r \subset \calE^{\star}_\mu \subset T^{\star}_p M \cong \R^k$ is a (normal) coordinate chart for $M$. Figure \ref{fig_geodesics_sphere_noise} illustrates the effect of noise on the geodesics, and therefore on the principal submanifold, for points distributed around the unit sphere. In the noiseless case, Figure \ref{fig_geodesics_sphere_noise} \textit{a)}, the computed geodesic paths are identical to the exact Riemannian geodesics on the sphere, up to numerical error, and the resulting principal submanifold is thus identical to the sphere (the mean norm of each generated point is $0.9992$ with standard deviation $0.0014$). In Figure \ref{fig_geodesics_sphere_noise} \textit{b)} the observations on the sphere have been added isotropic Gaussian noise in $\R^3$ with marginal standard deviation $\sigma = 0.1$. In this case the geodesics still evolve very close to the sphere (the mean norm of each generated point is $1.0299$ with standard deviation $0.0162$), but they start to cross after some integration steps, so that the manifold property of $M^{k}_{\mu}(r)$ seems to hold for a smaller value of the radius $r$ compared to the noiseless case.

\subsubsection{Relation to principal flows}\label{sect_rel_to_principal_flows}

We end this subsection with a discussion on the relation between a principal submanifold for $k = 1$ and the principal flow, described in \cite{panaretos2014principal}. For $k = 1$, integrating the Hamiltonian equations (\ref{eq_SR_ham}) yields the flow of the first eigenvector field $e_1$ starting from $\mu$. This is called the \textit{principal flow} in \cite{panaretos2014principal}, but the methods differ in important ways. Firstly, the principal flow at $p$ is based on a second moment which is  centered around $p$, not at the local mean around $p$. The span of the first eigenvector of such an uncentered second moment will be 'orthogonal' to the point cloud when evaluated at points outside of it. This causes the principal flow to stray away from the observations if it reaches such a point. As opposed to this, the first eigenvector of the centered second moment stays tangential to the point cloud when evaluated outside of it, as illustrated by  the pink curve in Figure \ref{fig_ps_illustrations} \textit{a)}. This behaviour arguably makes it more stable, see simulation results in section \ref{sect_applications_manif_data} and Figure \ref{fig_sphere_curves}. Secondly, to handle the fact that eigenvectors are determined only up to their sign, the principal flow is computed by solving a variational problem and integrating an associated system of ODE's. This system of ODE's has to be integrated for a range of candidate values of a Lagrange multiplier, in the end choosing the value for which the corresponding curve minimizes an energy functional. As opposed to this, we formulate the problem as a Hamiltonian system of ODE's which is invariant to the sign of the vector field (only the corresponding rank 1 subbundle matters), removing the need for the variational formulation and the ODE integration for multiple values of a Lagrange multiplier. It is this reformulation of principal flows as solutions to a set of geodesic (Hamiltonian) equations that also allows us to generalize the concept to higher dimensions.

\begin{figure}
     % \captionsetup{belowskip=0pt}
     \centering
     \begin{subfigure}[b]{0.55\textwidth}
         % \centering
         \includegraphics[trim=130 170 130 170,clip,width=\textwidth]{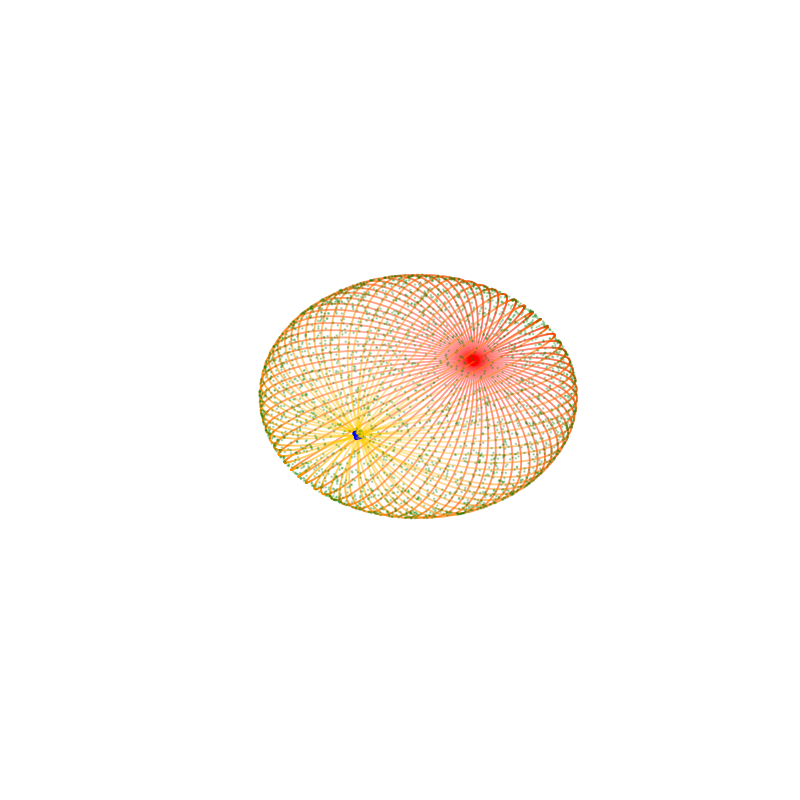}
         % \vspace{-0.9\intextsep}
         \caption{}
         \label{fig:y equals x}
     \end{subfigure}
     \hfill
     \begin{subfigure}[b]{0.49\textwidth}
         % \centering
         \includegraphics[trim=130 170 130 170,clip,width=\textwidth]{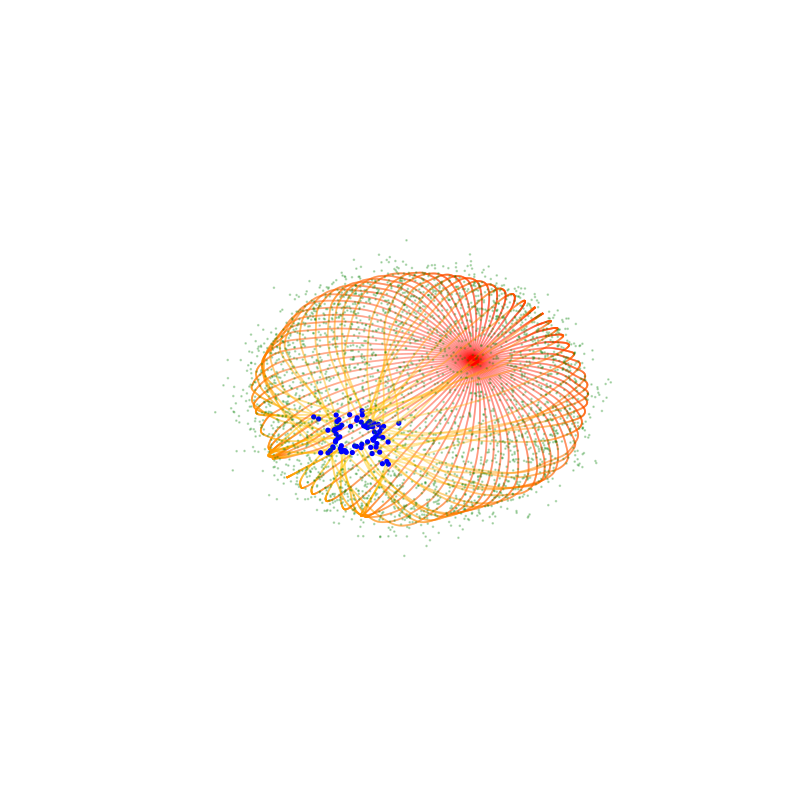}
         \caption{}
         \label{fig:three sin x}
     \end{subfigure}
     \caption{ Two illustrations of geodesics w.r.t. the sub-Riemannian metric induced by points clouds distributed around the unit sphere. On Figure \textit{a)} 2000 points (green points) $x_i$ are sampled uniformly on the sphere. On Figure \textit{b)} noise is added to the observations, which are now sampled from $y_i \cong N(x_i, I_3 \cdot \sigma)$, i.e. isotropic Gaussian distributions with marginal standard deviation $\sigma = 0.1$ (green points). On each figure, 75 geodesics with initial cotangents on a grid in the dual subbundle at the basepoint $\mu = (0,-1,0)$ are generated (these are the curves with a color gradient from red to yellow). The duration, and thus length of each geodesic is $T = \pi$, which theoretically corresponds to half a round on the unit sphere. The endpoint of each geodesic is marked by the blue dots. For points on the geodesics on Figure \textit{a)}, the mean norm is 0.9992 with standard deviation 0.0014 - thus the geodesics stay close to the true submanifold (the sphere). For points on the geodesics on Figure \textit{b)}, the mean norm is 1.0299 with standard deviation 0.0162 - thus the geodesics still stay close to the sphere, but they now deviate somewhat from great arcs, as illustrated by the endpoints not being exactly at the opposite pole.}
     \label{fig_geodesics_sphere_noise}
\end{figure}

\subsubsection{Projection to $M_\mu^k$}\label{sect_projection}

An observation $x_i \in \R^d$ can be projected to $M_{\mu}^k$ by $$\pi_{M_\mu^k}(x_i) = exp^{\calE}_\mu(\overline{\log}_\mu(x_i))$$ where $\overline{\log}_\mu(x_i)$ is a solution to (\ref{eq_log}) with search space $\mathbb{A} = \calE_\mu$. Alternatively, given a discrete representation $\widetilde{M_\mu^k}$ of $M_\mu^k$, computed using Algorithm \ref{algo_principal_submanifold}, one can use the discrete projection $\pi_{\widetilde{M_\mu^k}}(x_i) \defeq \argmin_{p \in \widetilde{M_\mu^k}} \Vert x_i - p \Vert$, which can be solved numerically as a Euclidean 1-nearest neighbours problem.

\begin{algorithm}[H] 
\caption{Point representation of a principal submanifold}
\label{algo_principal_submanifold}
\begin{algorithmic}[1]
\Require{
\Statex
\begin{itemize}
      \item \textit{Geometric parameters:} kernel range $\alpha \in (0, \infty)$, submanifold dimension $k \in \{1, \dots, d-1\}$, base point $\mu \in \R^d \setminus \mathcal{S}_{\alpha, k}$, radius $r>0$.
      \item \textit{Numerical parameters:} a number of geodesics $L \in \mathbb{N}$, the stepsize $\Delta > 0$.
      \vspace{2mm}
  \end{itemize}
} 
\Ensure{
\Statex
$sL + 1$ points in $M^{\mathcal{E}^{\alpha, k}}_{\mu}(r) \subset \R^d$, where $s = \lfloor{r/\Delta}\rfloor$ is the number of integration steps. 
\vspace{2mm}
}
\State \textit{Initialization}: Generate \textit{L} cotangents $\eta_{i}$ on the  $k$-dimensional unit sphere, $\eta_{i} \in \mathcal{S}^k \subset \left(\mathcal{E}^{k, \alpha}\right)^{\star} \cong \mathbb{R}^k, i = 1\dots L$.
   
   \vspace{2mm}
\Statex 
\For{{$i=1$ {\bfseries to} $L$}}
   \State \parbox[t]{313pt}{Integrate Hamiltonian equations (\ref{eq_SR_ham}) with initial condition $(\mu, \eta_i)$ over $s = \lfloor{r/\Delta}\rfloor$ steps of 
   stepsize $\Delta$.}
   \State Store the points along the trajectory; $p_{ij} = \exp^{\mathcal{E}}_{\mu}(j \Delta  \eta_i), j = 1\dots s$.
\EndFor
\vspace{2mm}
    \State \Return {Points $\left\{p_{ij} = \exp^{\mathcal{E}}_{\mu}(j \Delta  \eta_i) \middle\vert i = 1\dots N, j = 1\dots s\right\}$}
\end{algorithmic}
\end{algorithm}

\subsection{Representation of observations in $\R^k$ (Problem \textit{C})}\label{sect_Rk_repr}
The ball $B_r \subset \calE^\star_\mu \cong \R^k$ forms a coordinate chart for the principal submanifold, i.e. any point $p\in M^k_\mu(r)$ can be represented as $\bar{p} \defeq \exp^{-1}_{\mu}(p) \in \R^k$. It behaves like a socalled normal chart, in the sense that the SR distance between the base point $\mu$ and $p\in M^k_\mu$ is preserved, $d^{\mathcal{E}}(\mu, p) = \Vert \bar{p}\Vert$, while the distances between arbitrary points $p, q \in M_\mu^k$ are distorted in a way that depends on the curvature of $M^{\mathcal{\calE}}_{\mu}$. If $\{x_1, \dots, x_N\}$ are observations distributed around $M^k_\mu$, then the projections $\pi_{M_\mu^k}(x_i) \in M^k_\mu, i=1\dots N$, can be represented in this chart by solving the log problem (\ref{eq_log}) with $\mathbb{A} = \calE_\mu$,   yielding lower dimensional representations $\overline{\pi_{M_\mu^k}(x_i)} \defeq \overline{log}_\mu(\pi_{M_\mu^k}(x_i))\in\R^k, i =1..N$. Computing this is less complex than it looks; in fact, solving the projection problem (either the continuous or the discrete version, c.f. Section \ref{sect_projection}) already involves solving the log-problem, so computing a projection also yields the representation in $\R^k$.  See Figure \ref{fig_s_surface} and Section \ref{sect_s_surface} describing a 2D representation of the S-surface embedded in $\R^{100}$.

\subsection{Computing the SR distance between points (Problem \textit{A})}\label{sect_sr_dist}
As discussed in Section \ref{section_SR_exp}, we can combine Equations (\ref{eq_ham_vs_distance}) and (\ref{eq_log}) to approximate the SR distance between two points $x,y \in \R^d \setminus \mathcal{S}_{\alpha, k}$ by $$d^{\mathcal{E}}(x, y) \approx \sqrt{2 H(\overline{\log}_{x}(y))},$$ with log search space $\mathbb{A}=T_{\mu}^\star\R^d$. As mentioned, we cannot expect to find the exact SR distance, i.e. the length of the globally shortest curve joining $x$ and $y$, even in the case of a bracket-generating subbundle for which $d^\calE$ is in fact finite for all $x,y$. When the points are observations, i.e. $x,y \in \{x_i\}_{i =1, \dots, N}$, this might not be desirable either since the error in the log minimization problem (\ref{eq_log}) can be
% %% \begin{wrapfigure}[16]{l}{0.5\textwidth}
% \begin{wrapfigure}{l}{0.5\textwidth}
%   % \begin{center}
%   \vspace{-0.8\intextsep}%
%     \includegraphics[width=0.5\textwidth]{log_s_surface.png}
%   % \end{center}
%   % \vspace{-10pt}
%   \caption{Illustration of a computation of $\overline{\log}_x(y)$ based on observations (green dots) distributed around the S-surface. the base point, \textit{x}, is the blue dot and the target point, \textit{y}, is the pink dot. The red curve is the geodesic $t \mapsto \exp^{\calE}_x(t\cdot \overline{\log}_x(y)), t\in [0,1]$.}
%     % \label{fig_log_s_surface}
%   \label{fig_log_s_surface}
% \end{wrapfigure} 
%% \begin{wrapfigure}[16]{l}{0.5\textwidth}
\begin{figure}
  \begin{center}
  % \vspace{-0.8\intextsep}%
    \includegraphics[width=0.5\textwidth]{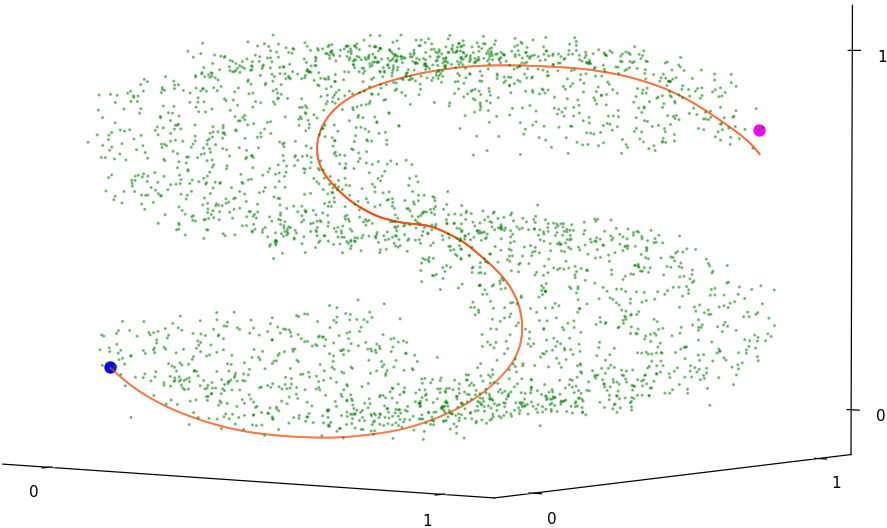}
  \end{center}
  % \vspace{-10pt}
  \caption{Illustration of a computation of $\overline{\log}_x(y)$ based on observations (green dots) distributed around the S-surface. the base point, \textit{x}, is the blue dot and the target point, \textit{y}, is the pink dot. The red curve is the geodesic $t \mapsto \exp^{\calE}_x(t\cdot \overline{\log}_x(y)), t\in [0,1]$.}
    % \label{fig_log_s_surface}
  \label{fig_log_s_surface}
\end{figure}

interpreted as an effect of random noise. 
See Section \ref{sect_learning_distance_metric} for a numerical evaluation of estimated distances $d^{\calE}$ based on a dataset in $\R^{50}$. 
Figure \ref{fig_log_s_surface} illustrates a computation of $\overline{\log}_x(y)$ based on a dataset distributed around the S-surface. The base point, \textit{x}, is the blue dot and the target point, \textit{y}, is the pink dot. The red curve is the geodesic $t \mapsto \exp^{\calE}_x(t\cdot \overline{\log}_x(y)), t\in [0,1]$, the length of which constitutes our estimate of the distance between $x$ and $y$. As expected, the endpoint $\exp^{\calE}_p(\overline{\log}_x(y))$ doesn't match $y$ exactly. On Figure \ref{fig_surface_recon}, the color gradient and concentric circles on the face illustrate the SR distance to the base point on the nose.

\subsection{Hyperparameters}

The kernel range $\alpha$ and the dimension \textit{k} are hyperparameters that are common to many methods and there is a significant body of literature about how to select them. See Appendix \ref{app_hyperparam} for our comments and references. Regarding the base point $\mu \in \R^d$ of a principal submanifold, we suggest to use a local mean around a well-chosen observation $x_{0}$. Which particular $x_{0}$ will be application specific, but a general purpose option is a within-sample \textit{Fréchet mean}, $$\hat{\mu} \in \underset{\mu \in \{x_i\}_{i=1..N}}{\argmin} \frac{1}{N}\sum_{i=1}^N d(\mu, x_i),$$ where $d$ is either the Euclidean distance or $d^{\calE}$ of the principal subbundle.

\begin{figure}[h!]
  \centering
  \includegraphics[scale=0.5]{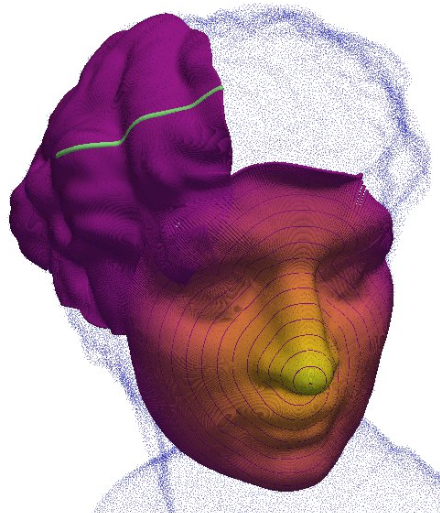}
  \caption{Illustration of the experiment described in Section \ref{sect_surface_recon}. Two principal submanifolds recontructing the 'head sculpture' surface from a noisy point cloud (blue points). One submanifold is centered approximately at the tip of the nose, the other is centered at the left end of the green line. The figure shows the raw points generated by Algorithm \ref{algo_principal_submanifold} - no subsequent processing, apart from coloring, has been applied. The skewed circles on the face are geodesic balls, i.e. points on the same circle have the same sub-Riemannian distance to the base point. Likewise, the colors of the face depends on the SR distance to the base point at the tip, a lighter color signifies shorter distance. The green line on the top submanifold highlights a single SR geodesic - each of the two submanifolds consists of $L=2500$ such geodesics.}
  \label{fig_surface_recon}
\end{figure}

%\begin{figure}[h]\label{fig_surface_recon}
%\centering
%\includegraphics[scale=0.35]{two_charts_closeup.png}
%\caption{Two principal submanifolds (green and yellow surfaces) recontructing the 'head sculpture' surface from a noisy point cloud (purple points). Each submanifold is a chart on the surface. The green and yellow parts show the points generated by Algorithm \ref{algo_principal_submanifold} - no subsequent processing has been applied.}
%\end{figure}

\section{Generalization to observations on a Riemannian manifold}\label{sect_manif_valued_case}

In this section, we generalize the framework of principal subbundles to the setting where the observations are points on an a priori known Riemannian manifold. A numerical application of the method to such data is presented in Section \ref{sect_applications_manif_data}. These two sections assume a deeper knowledge of differential geometry than elsewhere, but they can be skipped without loss of continuity by the reader who wish to focus on the case of Euclidean valued data. It turns out that the formulation of principal subbundles for Euclidean valued data, given above, is based only on operations that generalize naturally to the setting of manifold valued data, as we show below. 

\subsection{Context: geometric statistics}

We now assume that $\{x_i\}_{i=1\dots N}$ are points on an a priori known smooth manifold $\mathcal{N}$ of dimension $d < \infty$, 
equipped with a known Riemannian metric $h$. This is a generalization of the theory presented above, where $\mathcal{N}$ was $\R^d$ and $h$ was the Euclidean metric. Our aim is now to find a lower dimensional geometric structure (e.g. a submanifold) \textit{within} this given manifold $\mathcal{N}$. 

The field of statistics and machine learning for manifold-valued data is called \textit{geometric statistics} \cite{pennec2019riemannian}. An intuitive example of such data is observations on a surface in $\R^3$, such as the sphere. More abstract examples are \textit{shapes} represented as sets of landmarks, e.g. in Kendall's shape space (see \cite{kendall1984shape} and \cite{huckemann2010intrinsic} for an application) or an LDDMM landmark manifold \cite{younes2010shapes}. Other examples are provided in the field of directional statistics \cite{pewsey2021recent} and image processing via the manifold of SPD matrices (e.g. \cite{pennec2019riemannian}, Chapter 3).

Within the field of geometric statistics, several methods have been proposed to find a lower dimensional submanifold $\widehat{M}\subset \mathcal{N}$ approximating the observations in $\mathcal{N}$. Important methods are Principal Geodesic Analysis (PGA) \cite{fletcher2004principal}, Principal Nested Spheres \cite{jung2010generalized}) and Barycentric Subspace Analysis \cite{pennec2018barycentric}. A basic method is \textit{tangent PCA}, which consists of mapping the observations to a tangent space at a chosen base point $\mu \in \mathcal{N}$ via the Riemannian logarithm and performing Euclidean PCA in this linear representation. This method is not sensitive to the curvature of neither $\mathcal{N}$ nor of the dataset. Tangent PCA can be seen as a linear approximation of PGA (as discussed in \cite{fletcher2004principal}), a method which is more sensitive to the curvature of $\mathcal{N}$ but still not sensitive to the curvature of the dataset, as we now describe. Let $\mu \in \mathcal{N}$ be a well-chosen base point, e.g. the Fréchet mean. The collection of geodesics initialized by tangent vectors in a \textit{k}-dimensional subspace $\Delta_k \subset T_{\mu}\mathcal{N}$ form a \textit{k}-dimensional submanifold of $\mathcal{N}$, given as $\exp^h_\mu(\Delta_k)$, where $\exp^h$ is the Riemannian exponential map of $(\mathcal{N},h)$. Starting from $k=1$ and adding subsequent dimensions one at a time, the optimal subspace $\Delta_k$ at each step is defined to be the minimizer of the geodesic distance from $\exp^h_\mu(\Delta_k) \subset \mathcal{N}$ to the observations. This optimization is computationally intensive, to the point of being infeasible for even fairly simple manifolds and datasets - no publicly available implementation of PGA exists to this date (for work in this direction, see e.g. \cite{sommer2010manifold}). Furthermore, geodesics of the ambient space $(\mathcal{N},h)$, which forms the approximating submanifold $\exp^h_\mu(\Delta_k)$, is a relatively inflexible family of curves - they are the generalization of straight lines to a manifold. 

A principal submanifold constructed from a principal subbundle on $\mathcal{N}$ can be seen as a locally data-adaptive combination of tangent PCA and PGA. We compute local tangent PCA's to construct the principal subbundle $\calE_\alpha$ of the tangent bundle $T\mathcal{N}$. This determines a  data-dependent sub-Riemannian metric and thus sub-Riemannian geodesics on $\mathcal{N}$, with which we can approximate the data. That is, compared to PGA, the geodesics forming the principal submanifold are not those of the ambient Riemannian manifold $(\mathcal{N}, g)$, but those of an estimated sub-Riemannian structure on $\mathcal{N}$. Our approximating submanifold is $\exp^\calE_\mu(\Delta_k)\subset \mathcal{N}$, similar to PGA, except that the exponential is now the sub-Riemannian exponential determined by the principal subbundle and $\Delta_k$ is the metric dual of $\calE_\mu$, the principal subbundle at $\mu$. Note that by doing local PCA's (i.e. solving many simple, local least-squared-error problems) we remove the need for the expensive 'global' optimization for the subspace $\Delta_k$.

\subsection{Sub-Riemannian structures on a general smooth manifold}

This section introduces sub-Riemannian geometry on a smooth manifold $\mathcal{N}$ 
of dimension $d$, 
not necessarily $\R^d$. A rank \textit{k} sub-Riemannian structure on $\mathcal{N}$ is determined by a rank \textit{k} subbundle $\calD$ of $T\mathcal{N}$ and a metric tensor $g$ on $\calD$. We will assume that the sub-Riemannian metric tensor $g$ is the restriction $h\vert_\calD$ of a given Riemannian metric tensor $h$ on $T\mathcal{N}$ to $\calD$, i.e. $g_x(u,v) = h_x(u,v)$ for all $x \in \mathcal{N}$ and $u,v \in \calD_x$. The pair $(\calD, g)$ is equivalent to a rank \textit{k} cometric tensor $g^{\star}$ on $T^\star \mathcal{N}$. The triple $(\mathcal{N}, g, \calD)$, or equivalently the pair $(\mathcal{N}, g^\star)$, is called a sub-Riemannian manifold. The version of sub-Riemannian geometry we described and used in the previous sections corresponds to $\mathcal{N} = \R^d$ and the ambient Riemannian metric $h$ being the Euclidean metric.

In this general setting, a curve $\gamma : [0, T] \to \mathcal{N}$ is still called horizontal if its velocities satisfy $\dot{\gamma}_t \in \calD_{\gamma_t} \subset T_{\gamma_t}\mathcal{N}$ for all $t \in [0,T]$. And this again induces the Carnot-Carath\'eodory distance metric $d^\calD$ (equation \ref{eq_carnot_carath_metric}) on $\mathcal{N}$. The discussion in Section \ref{sect_integrability} about integrability and foliations carries over directly; the subbundle $\calD$ partitions $\mathcal{N}$ into a foliation of submanifolds of dimension $\Lie \calD$, and the distance metric $d^\calD(x,y)$ is finite only between points on the same leaf. The Hamiltonian equations, $\exp$ and $\log$ are also defined exactly as in Section \ref{sect_sr_geom_eucl}, and the relationship between the sub-Riemannian distance and the Hamiltonian (Eq. (\ref{eq_ham_vs_distance})) still holds. One difference from the previous Euclidean setting, however, is that the cometric cannot be expressed as a projection matrix, as we did in Equation (\ref{eq_cometric_eucl}).
%\todo{yes, in coordinates? But the expression for the matrix is slightly more complicated since the ambient metric is now arbitrary,,,?}
Therefore it is more convenient to represent the Hamiltonian in the following equivalent way (see \cite{agrachev2019comprehensive}, Proposition 4.22 for a derivation,
%\todo{,,, er mit bevis relevant her?})
\begin{align}\label{eq_ham_general_manif}
    %% H : T\mathcal{N} \to \R_{\geq 0} : H(x,\eta) = \frac{1}{2} \sum_{i=1}^N \left(\eta(f_i(x))\right)^2,
    H : T\mathcal{N} \to \R_{\geq 0} : H(x,\eta) = \frac{1}{2} \sum_{i=1}^d \left(\eta(f_i(x))\right)^2,
\end{align}
where $\{f_i\}_{i=1..k}$ is an orthonormal frame for $\calD$ w.r.t. $g$ and $\eta(f_i(x))$ denotes the cotangent $\eta \in T^{\star}_x \mathcal{N}$ evaluated at the tangent $f_i(x) \in T_x \mathcal{N}$. The derivatives of the Hamiltonian that enter into the Hamiltonian equations can be expanded in a way that is suitable for implementation (see Equation (4.38) in \cite{agrachev2019comprehensive}).

To construct a $k$-dimensional submanifold from a $k$-dimensional \textit{non-integrable} subbundle we still need a result such as Proposition \ref{prop_sr_exp_diffeo}, which luckily holds in this general setting - cf. the proof in Appendix \ref{app_proof_exp_local_diffeo}. 
The result carries over verbatim, with the dual subbundle now being the dual w.r.t. our (general) Riemannian metric $h$ on $\mathcal{N}$, i.e. $$\calD^\star_x \defeq \{h_x(v, \cdot) \vert v\in\calD_x\} \subset T_x^{\star} \mathcal{N}.$$

\subsection{Principal subbundles on a Riemannian manifold}

We now generalize local PCA to the setting of observations on a Riemannian manifold. In this setting, local PCA is exchanged for local \textit{tangent} PCA, by which we mean the extraction of eigenvectors from the following second moment.

\begin{definition}[Non-centered weighted tangent second moment on a Riemannian manifold]\label{def_weighted_2nd_moment_manif} 
Let $\{x_1,\dots,x_N\}$ be observations on a Riemannian manifold $(\mathcal{N}, h)$. Let $K_{\alpha} : \mathbb{R}_{\geq 0} \to \mathbb{R}_{\geq 0}$ be a smooth, decaying kernel function with range parameter $\alpha > 0$. 
At a point $p\in \mathcal{N}$, we denote by $\log^h_p(x_i)$
%, $i=1\dots N$ 
the Riemannian $\log$ of the observation point $x_i$  w.r.t. metric $h$. 
The \emph{weighted tangent second moment} is defined by
\begin{align*}
\Sigma_{\alpha}(p) = 
\sum_{i=1}^{N} w_i(p) \left(\log^h_{p}(x_i) \otimes \log^h_{p}(x_i)\right),
\end{align*}
% \noindent with weight functions 
% \begin{align}\label{eq_weight_functions_manif}
%     w_i &: \mathcal{N} \to \R_{\geq 0} \\
%     &: p \mapsto K_\alpha(\Vert \log^h_{p}(x_i) \Vert_{p})
% \end{align}  
\noindent with normalized weight functions
\begin{equation}\label{eq_weight_functions_manif}
    w_i \hspace{1mm}: \hspace{1mm} \mathcal{N} \to  \R_{\geq 0} \hspace{1mm} : \hspace{1mm}  p\mapsto w_i(p) = \frac{K_\alpha(\Vert \log^h_{p}(x_i) \Vert_{p})}{\sum_{j=1}^N K_\alpha(\Vert \log^h_{p}(x_j) \Vert_{p})}.
\end{equation} 
\end{definition}

\begin{remark}
    Recall that $\Vert \log^h_{p}(x_i) \Vert_{p}= d^h(p, x_i)$ since the length of the shortest geodesic from $p$ to $x_i$ is precisely the length of the vector in $T_p \mathcal{M}$ that exponentiates to $x_i$.  
\end{remark}

For any $v,u \in T_p \mathcal{N}$, the tensor product $v \otimes u$ can be identified with a linear map on $T_p \mathcal{N}$ (an endomorphism), whose coordinate representation is a $d \times d$ matrix, see Lemma \ref{lemTensorCoordinates}. There is some vagueness about the exact form of this coordinate representation in the geometric statistics literature, so we give a detailed proof in Appendix \ref{app_second_moment_as_tensor}.

 \begin{restatable}{lemma}{lemTensorCoordinates}\label{lemTensorCoordinates}
        Let $(\mathcal{N},h)$ be a Riemannian manifold, and $u,v\in T_p\mathcal{N}$. Given a choice of basis for $T_p \mathcal{N}$, the tensor $v \otimes u \in T_p{\mathcal{N}}\otimes T_p{\mathcal{N}}$ can be expressed in coordinates as
    \begin{align}\label{eq_matrix_repr}
        vu^T h_p \in \R^{d\times d},
    \end{align}
    where $u,v \in \R^{d\times 1}$ are the vectors  and $h_p \in \R^{d\times d}$ is the Riemannian metric represented w.r.t. the chosen basis. 
\end{restatable}

In various sources, the term $h_p$ in Eq. (\ref{eq_matrix_repr}) is omitted without explanation. We stress that this is only correct if the chosen coordinate representation of the metric is the identity matrix, e.g. if the chart is a normal chart - which is not necessarily the case in numerical computations. Sometimes, this is ensured by changing the basis to an orthonormal one, found by e.g. Cholesky decomposition of the cometric, before computing $vu^T$.
This is, however, much more expensive than simply using the general, basis independent, expression (\ref{eq_matrix_repr}).
Thus, when computing the tangent second moment matrix (e.g. when computing tangent PCA), the covariance matrix w.r.t. some arbitrary basis \textit{a} should be computed as
\begin{equation*}
\left[\Sigma_{\alpha}(p)\right]_a = 
\sum_{i=1}^{N} w_i(p) \left[\log^h_{p}(x_i)\right]_a \left( \left[\log^h_{p}(x_i)\right]_a \right)^T [h_p]_a.
\end{equation*}
    
As in the case of Euclidean valued data, we want the principal subbundle of $T\mathcal{N}$ to be based on local PCA's centered around local means. For that purpose, the principal subbundle subspace at point \textit{p} will be based on the eigendecomposition of the weighted second moment at the weighted mean $m(p)$ defined below:
\begin{definition}[Weighted tangent mean map on a Riemannian manifold]\label{def_weighted_tangent_mean} 
Let $\{x_1,\dots,x_N\}\subset \mathcal{N}$ be observations on a Riemannian manifold $(\mathcal{N}, h)$, let the normalized weight functions $w_i$
%, $i =1 \dots N$, 
be defined as in (\ref{eq_weight_functions_manif}), and let $\exp^h_p$ be the Riemannian exponential map at $p$ w.r.t. metric $h$. The \emph{weighted tangent mean map} is defined by
\begin{equation}\label{eq_weighted_mean}
   m\: \hspace{1mm}: \hspace{1mm} \mathcal{N} \to  \mathcal{N} \hspace{1mm} : \hspace{1mm}  p \mapsto m(p) = \exp^h_p\left(
    %\frac{1}{\sum_{i=1}^N  w_i(p)} 
    \sum_{i=1}^N w_i(p) \log^h_p(x_i)\right).
\end{equation}

\end{definition}

The eigenvectors of $\Sigma_{\alpha}(m(p))$ belong to the tangent space at $m(p)$, not the tangent space at $p$. 
Thus, the extracted eigenvectors needs to be mapped back to the tangent space at $p$, which we do by parallel transport, as described in the definition below. 

The principal subbundle on $(\mathcal{N}, h)$ can only be defined at points  $p$ s.t. both $p$ and $m(p)$ is in the cut locus of every observation and of each other, since we need to compute the corresponding logarithms. We therefore define the set of singular points as follows,

\begin{align}\label{eq_singular_set_manif}
S^{'}_{\alpha, k} = \bigg\{p \in \mathcal{N} \hspace{0.8mm} \vert \hspace{0.8mm} &p, m(p) \in \bigcup_{q \in \{x_1, \dots, x_N, p\}}\text{Cut}(q) \hspace{3mm}\text{or}\hspace{3mm} \lambda_k(m(p)) = \lambda_{k+1}(m(p))\bigg\},
\end{align}

\vspace{2mm}

\noindent where $\lambda_i(m(p))$ is the $i$'th eigenvalue of $\Sigma_{\alpha}(m(p))$ of Definition \ref{def_weighted_2nd_moment_manif}. 

\begin{definition}[Principal subbundle on a Riemannian manifold]\label{def_PS_manif} Let $\lambda_1(q) \geq \dots \geq \lambda_d(q)$ be the eigenvalues of $\Sigma_{\alpha}(q)$, at $q \in \mathcal{N}$, with associated eigenvectors $e_1(q),\dots, e_d(q)$. 
Let $\Pi_x^y(v)$ denote parallel transport of $v \in T_x \mathcal{N}$ to $T_y \mathcal{N}$ along the length-minimizing geodesic between $x$ and $y$. Then the \emph{principal subbundle} $\mathcal{E}^{k, \alpha} \subset T \mathcal{N}$ is defined as
\begin{align*}
\mathcal{E}^{k, \alpha} = \biggl\{(p,v) \hspace{1mm}\vert\hspace{1mm} &p \in \mathcal{N} \setminus \mathcal{S}^{'}_{\alpha, k}, \\
&v\in \text{span}\left\{\Pi_{m(p)}^p e_1(m(p)),\dots, \Pi_{m(p)}^p e_k(m(p))\right\}\biggr\}    
\end{align*}
\end{definition}
\begin{remark}
    If $(\mathcal{N}, h)$ is Euclidean space, the above definition reduces to the Euclidean Definition \ref{def_PS} since $\left(\log^h_{q}(x_i) \otimes \log^h_{q}(x_i)\right) = (x_i - q)(x_i - q)^T$ and $\Pi_q^p$ is the identity map for $q \in \R^d$. 
\end{remark}
\begin{remark}\label{rem_PS_manif_approx_tangent_space} The above construction of the subbundle subspace at $p$ can be approximated by using the Euclidean definition in the tangent space at $p$, i.e. by letting $\mathcal{E}^{k, \alpha}_p$ at $p\in \mathcal{N}$ be the span of eigenvectors of $\Sigma_{\alpha}(0)$ computed from vectors $\{\log^h_p(x_i)\}_{i=1\dots N}\subset T_p \mathcal{N} \cong \R^d$, where $\Sigma_{\alpha}$ is the Euclidean second moment from Definition \ref{def_weighted_2nd_moment}. In this way, only $N$ log's have to be computed, instead of $2N$ (see Algorithm \ref{algo_principal_subbundle_general}), and the parallel transport operation is omitted. Note that the experiments in Section \ref{sect_applications_manif_data} uses Definition \ref{def_PS_manif}, not the described approximation. 
\end{remark}

Algorithm \ref{algo_principal_subbundle_general} describes how to compute the principal subbundle from data on a Riemannian manifold $(\mathcal{N},h)$. 

As in the Euclidean case, we prove that the principal subbundle on $(\mathcal{N}, h)$ is smooth at all points where it is defined.

\begin{restatable}{proposition}{propSmoothPsManif}
    The principal subbundle, defined on $\mathcal{N} \setminus S^{'}_{\alpha, k}$, is smooth. 
\end{restatable}

\begin{algorithm}[!h] 
\caption{Computing the principal subbundle at a point on a Riemannian manifold}\label{algo_principal_subbundle_general}
\begin{algorithmic}[1]
\Require{
\Statex
\begin{itemize}
       \item Observations $\{x_i\}_{i=1\dots N}$ on a Riemannian manifold $\mathcal{N}$ of dimension $d$, and a point $p \in \mathcal{N}$, satisfying $m(p) \in \mathcal{N} \setminus \mathcal{S}_{\alpha, k}$, at which to compute the subbundle subspace.
       \item Parameters $\alpha \in (0, \infty)$ (range of the kernel),  $k \in \{1, \dots, d-1\}$ (dimension of the subspace).
       \vspace{2mm}
  \end{itemize}
} 
\Ensure{
\Statex
 A set of vectors spanning the principal subbundle subspace at \textit{p}, $\calE_p \subset T_p \mathcal{N}$. \vspace{2mm}
}
   
\Statex 
\For{{$i=1$ {\bfseries to} $N$}} 
   \State Compute the normalized weight % of observation $i$ at $p$, 
   $w_i(p) = \frac{K_\alpha(\Vert \log^h_{p}(x_i) \Vert_{p})}{\sum_{j=1}^N K_\alpha(\Vert \log^h_{p}(x_j) \Vert_{p})}.$
\EndFor
\vspace{2mm}

\State Compute the weighted mean $m(p) \in \mathcal{N}$ around $p$ (Eq. (\ref{eq_weighted_mean})).

\Statex 
\For{{$i=1$ {\bfseries to} $N$}} 
   \State Compute the normalized weight  at $m(p)$, 
   \[ w_i(m(p)) = \frac{K_\alpha(\Vert \log^h_{m(p)}(x_i)  \Vert_{m(p)})}{\sum_{j=1}^N K_\alpha(\Vert \log^h_{m(p)}(x_j)  \Vert_{m(p)})}. \]
\EndFor

   \vspace{2mm}
   
   \State Compute the weighted second moment at the weighted mean, 
    \[\Sigma_{\alpha}(m(p)) = \sum_{i=1}^{N} w_i(m(p)) \left(\log^h_{m(p)}(x_i) \log^h_{m(p)}(x_i)^T h_p\right).\] 
    % \vspace{2mm}

   \State Eigendecompose $\Sigma_{\alpha}(m(p))$ and select the $k$ eigenvectors $\{e_1, \dots,e_k\}$ with the largest $k$ eigenvalues. 

   \vspace{2mm}

   \State Parallel transport each eigenvector from $T_{m(p)} \mathcal{N}$ to $T_p \mathcal{N}$ along the length-minimizing geodesic between $m(p)$ and $p$, yielding $e_i^{\star} \defeq \Pi_{m(p)}^p e_i \in T_p \mathcal{N}$.

\vspace{2mm}
    \State \Return {$\{e_1^\star, \dots, e_k^\star\}$. }
\end{algorithmic}
\end{algorithm}

\subsection{Computing with a principal subbundle on a Riemannian manifold}

Given a dataset $\{x_1,\dots,x_N\}\subset \mathcal{N}$, the associated principal subbundle $\calE$ determines a sub-Riemannian structure on $\mathcal{N}$, namely $(\mathcal{N}, h\vert_\calE, \calE)$. Using this structure, we can integrate the associated sub-Riemannian Hamiltonian equations in the same way as described in section \ref{sect_computing_geos_eucl}, except that we use the expression (\ref{eq_ham_general_manif}) for the Hamiltonian. This gives us sub-Riemannian exponential and logarithmic maps on $\mathcal{N}$, so that problems \textit{A}, \textit{B} and \textit{C} can be solved on a general Riemannian manifold, in exactly the same way as in the Euclidean case, described in sections \ref{sect_PS}-\ref{sect_sr_dist}.

A principal submanifold is computed in the same way as in the Euclidean case (Algorithm \ref{algo_principal_submanifold}). It assumes that 
we have a representation of the manifold in a chart. See the pseudocode for our implementation on the sphere in Appendix.
Due to the centering step, 
computing 
the subbundle at a point $p \in M$ requires solving the parallel transport equation and computing $2N$ log maps, $N$ logs between the observations and point $p$ (lines 1-3), and $N$ logs between the observations and the local mean around $p$ (lines 5-7). See remark \ref{rem_PS_manif_approx_tangent_space} for an approximation requiring only $N$ log computations and no parallel transport. The run time of the algorithm thus depends heavily on the run time of the log map, or an approximation thereof, on the given Riemannian manifold. Examples of manifolds with computationally cheap log maps are hyperspheres, Kendall shape space, Grassmann manifolds, SPD matrices. See the Python library Geomstats \cite{geomstats2020JLMR} for implementations of various manifolds including efficient log maps.

\section{Applications}\label{sect_applications}

We now demonstrate how principal subbundles  provide solutions to problems $A, B, C$, mentioned in the introduction. In particular, we reconstruct 2D submanifolds embedded in $\R^3$ and $\R^{100}$, respectively, and give a 2D tangent space representation of the latter. We furthermore evaluate a sub-Riemannian distance metric on $\R^{50}$ learned from observations  distributed around a $4$-dimensional sphere embedded in $\R^{50}$. In subsection \ref{sect_applications_manif_data} we compute a 1D principal submanifold approximating data on the sphere (a Riemannian manifold).

\subsection{Surface reconstruction in $\R^3$ (problem \textit{B})}\label{sect_surface_recon}

We reconstruct a 2D 
surface, the 'head sculpture', based on a point cloud contained in the surface reconstruction benchmark dataset from \cite{surveySurfacReconstructionHuang}. According to the classification in \cite{surveySurfacReconstructionHuang}, the surface is of complexity level 2 out of 3, and the point cloud has been added noise of level 2 out of 3, see \cite{surveySurfacReconstructionHuang} for details. Note that the evaluations in the benchmarking paper was made after a preliminary denoising step, whereas our reconstruction was done on the raw point cloud. This is to illustrate the potential use of principal submanifolds for denoising. The hyperparameters we use for the principal subbundle are $\alpha = 0.001$, and $k = 2$. See Appendix \ref{app_surface_recon_high_noise} for a reconstruction of the face using observations distorted by noise level 3 out of 3.  

Figure \ref{fig_surface_recon} shows two principal submanifolds reconstructing the head sculpture locally: one is based around the tip of the nose (radius $r=0.3$) and one at the top left side of the head ($r=0.25$). Both base points are computed as the kernel-weighted mean around a chosen observation. The numerical parameters in Algorithm \ref{algo_principal_submanifold}, determining the resolution, were $L = 2500$ (the number of geodesics) and $\Delta = 0.001$ (the integration stepsize).

A principal submanifold corresponds to a chart on the surface; in particular, a normal chart. It is a basic fact of differential geometry that a complicated surface such as the head sculpture cannot be covered by a single such chart. One therefore needs to reconstruct the surface based on multiple principal submanifolds corresponding to different base points; however, principal submanifolds based at different points might not overlap in a smooth way due to noise. To construct a smooth surface covering the whole area, we thus need a scheme for combining different principal submanifolds $M^{\mathcal{E}}_{\mu_1},M^{\mathcal{E}}_{\mu_2}, \dots$. Many such schemes are conceivable. In appendix \ref{app_combining_PS}, we propose one that combines submanifolds by weighing points according to their sub-Riemannian distance to a set of nearest base points. The discrepancy between submanifolds in the areas of overlap depends on the level of noise. In the experiment shown in Figure \ref{fig_surface_recon} we did not find it necessary to use a weighing scheme - see Appendix \ref{app_overlapping} for a close-up illustration of the overlap.

%[,,,]The figure illustrates how the resolution is higher close to the base point and diminishes as the geodesics diverge.

\subsection{Unfolding the S-surface in $\R^{100}$ (problem \textit{C})}\label{sect_s_surface}

In this experiment, we demonstrate the use of principal subbundles to contruct a representation of $\R^d$-valued data in $\R^k$, $k < d$. Let $y_i \defeq ((y_i)_1, (y_i)_2, (y_i)_3)^T \in \R^3$, $i=1..3000$,
%[,,,]
 be points on the S-surface, scaled such that its height, width and depth is 1. We embed each point in $\R^d$, $d = 100$, by adding zeros, $\tilde{y_i} = \left((y_i)_1, (y_i)_2, (y_i)_3, 0, \dots, 0\right)^T$. The observations are then generated by adding Gaussian noise,  $x_i \sim N(\tilde{y_i}, \sigma^2 I_d) \in \R^d$ for $\sigma = 0.025$.

The upper part of Figure \ref{fig_s_surface} shows the observations $\{x_i\}_{i=1..N}$ and an approximating principal submanifold, projected to $\R^3$ for the purpose of visualization. The base point of the principal submanifold is the local mean around the within-sample Fréchet mean w.r.t. Euclidean distance, $\mu = (0.47, 0.47, 0.49)$. The lower part of Figure \ref{fig_s_surface} shows the \textit{log} representation of the observations in $\calE_\mu^{\star} \cong T_{\mu}M_\alpha^k$. The kernel range is $\alpha =0.01$ and the rank is $k = 2$.

\subsection{Learning a distance metric on $\R^{50}$ (problem \textit{A})}\label{sect_learning_distance_metric}

We sample $N=10000$
points, $\{y_i\}_{i=1..N}$, uniformly on the \textit{k}-dimensional unit sphere embedded in $\R^d$, for $k = 4$,
$d=50$.
 For each of these points $y_i \in \R^d$ we generate an observation $x_i \in \R^d$ by adding $d$-dimensional Gaussian noise, $x_i \sim N(y_i, \sigma I_d)$, where $\sigma=0.01$.

We generate 20 such data sets with associated principal subbundles $\calE_j$, $j=1..20$.
For each data set we compute the SR distance $d^{\calE_j}(p,q), j=1\dots 20$, where $p=(1,0,\dots,0)\in\R^d$  and $q= (- \sqrt{1/2}, - \sqrt{1/2}, 0, \dots, 0)\in\R^d$. We find the mean, $\mu_0$, and standard deviation, $\sigma_0$, of these 20 computed distances to be $\mu_0 = \frac{1}{20} \sum_{j=1}^{20} d^{\calE_j}(p,q) = \frac{3}{4}\pi + 0.023$, $\sigma_0 = 0.025$. This result shows that the learned distances are close to true distance, $d^{\mathbb{S}^4}(p,q) = \frac{3}{4}\pi$, on the $4$-dimensional sphere.

\subsection{Curve approximation on the sphere}\label{sect_applications_manif_data}

In this experiment we randomly generate $20$ datasets, each with $N = 100$ points distributed around a random curve on the sphere, $\mathcal{S}^2$. The random curves are generated as follows. A 4'th order polynomial 
\begin{align}\label{eq_polynomial}
f : \R \to \R : t \mapsto  (t - a_1)(t - a_2)(t - a_3)(t - a_4)    
\end{align}
is generated by sampling roots $a_1, a_2$ from a uniform distribution on $(-1,0)$, and roots $a_3, a_4$ from a uniform distribution on $(0,1)$. Using two such intervals yields polynomials with more complex curvature. The graph of the polynomial, $P \defeq \{t, f(t) \vert t \in [-1,1] \}$, is considered a subset of $T_{p_0} \calS^2$ and mapped to $\mathcal{S}^2$ by the Riemannian exponential, $\exp_{p_0}$, where $p_0 = (0, 0, 1)$ is the north pole (in extrinsic coordinates).
Let $\{t_i\}_{i=1..N} \subset [-1,1]$ be 100 evenly spaced points. Let $z_i = \exp_{p_0}((t_i, f(t_i)))$, $i=1\dots N$, be points on the curve on $\mathbb{S}^2$. The noisy observations are generated as $x_i = \exp_{z_i}(v_i)$, where $v_{i} \sim \text{N}(0, I_2 \cdot \sigma)$, a 2D isotropic Gaussian with marginal variance $\sigma$, assuming a representation of $T_{z_i}\mathbb{S}^2$ in an orthonormal basis. In our experiments we used $\sigma = 5\cdot 10^{-4}$. Note8 that the resulting observations on $\mathbb{S}^2$ are non-uniformly sampled along the curve (making the problem more difficult). See Figure \ref{fig_sphere_curves} for an example of such a randomly generated dataset.

For each randomly generated dataset we estimate a base point as the within-sample Fréchet mean w.r.t. the geodesic distance on the sphere. We use as kernel function a Gaussian density with standard deviation $\alpha = 0.045$. This value is hand picked since our aim is to compare the performance of different methods disregarding uncertainty due to estimation of hyperparameters. Using this kernel function, we compute 3 curve approximations of the data set. Firstly, we compute the principal submanifold using Algorithm \ref{algo_principal_submanifold}.  Secondly, we compute the Principal submanifold without the centering and parallel transport step, i.e. the Principal flow \cite{panaretos2014principal}. Thirdly, we compute as baseline model the first principal geodesic from tangent PCA. For each approximation we compute the sum of squared errors (SSE), where the errors are measured by the length of the geodesic joining observation $x_i$ and its geodesic projection to the given curve. Figure \ref{fig_sphere_curves} shows an example data set and its 3 curve approximations. Figure \ref{fig_sphere_boxplots} shows boxplots summarizing the 20 SSE's computed for each approximation method. 

The SSE's and visual inspection of the corresponding plots shows that the centered version of the Principal submanifold is significantly more stable than the uncentered version (the principal flow). The uncentered version tends to stray away from the data when it reaches positions slightly outside of the point cloud. This is as expected, c.f. our discussion in Section \ref{sect_PS}. The principal geodesic has the highest SSE, as expected for this type of data that is distributed around a curve with relatively high curvature.

\begin{figure}[h!]
\centering
\begin{minipage}[t]{.5\textwidth}
  \centering
  \includegraphics[scale=0.35]{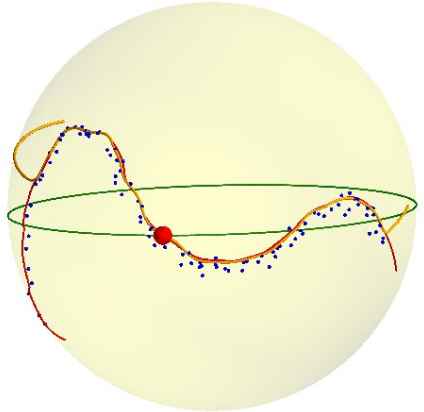}
  \caption{Three curves approximating a set of observations (blue points) on the sphere $\mathbb{S}^2$. The green curve is the principal geodesic computed by tangent PCA centered at the red point. The yellow curve is the principal submanifold based on a non-centered second moment (i.e. it is a principal flow). The red curve is the principal submanifold based on our proposed centered second moment. The base point of both principal submanifolds is the red point.}
  \label{fig_sphere_curves}
\end{minipage}%
\begin{minipage}[t]{.5\textwidth}
  \centering
  \includegraphics[width=1\linewidth]{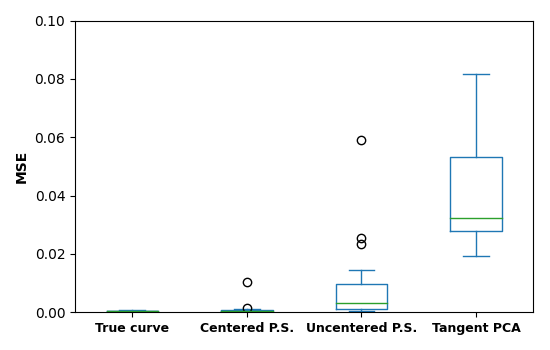}
  \captionsetup{width=.9\linewidth}
  \captionof{figure}{A box plot comparing the sum of squared errors (SSE), measured w.r.t. geodesic distance on the sphere, for each curve approximation. The 'True curve' label refers to the SSE for the curve $t\mapsto \exp_{p_0}\circ (t, f(t))$, described in Section \ref{sect_applications_manif_data}, from which noisy samples are generated. The other labels refers to the curves described in the caption of Figure \ref{fig_sphere_curves}, with 'P.S.' abbreviating principal submanifold.}
  \label{fig_sphere_boxplots}
\end{minipage}
\end{figure}

\section{Discussion and further work}

We have introduced the idea of modelling a data set $\{x_1,\dots,x_N\} \subset \R^d$ by a tangent subbundle consisting of affine subspaces of $\R^d$, and the sub-Riemannian geometry that it induces. We have demonstrated that geodesics w.r.t. this sub-Riemannian structure can be used to solve a number of important problems in statistics and machine learning, such as: reconstruction of submanifolds approximating the observations, finding lower dimensional representations and computing geometry-aware distances. Furthermore, we have shown that the framework generalizes to datasets on a given Riemannian manifold.

It can be considered a drawback of the framework that the point cloud must be relatively well connected, in the sense of not having large 'holes' or disconnected parts, relative to the kernel range. However, we conjecture that this can be somewhat alleviated by introducing a position-dependent range parameter.

\section*{Acknowledgements}

M.A., J.B. and X.P. are supported by the European Research Council (ERC) under the EU Horizon 2020 research and innovation program (grant agreement G-Statistics No. 786854). S.S. is partly supported by Novo Nordisk Foundation grant NNF18OC0052000 as well as Villum Foundation research grant 40582 and UCPH Data+ Strategy 2023 funds for interdisciplinary research. E.G. is supported by project GeoProCo from the Trond Mohn Foundation - Grant TMS2021STG02.

\printbibliography

\appendix

\section{Proofs}\label{app_proofs}

\subsection{Smoothness of the principal subbundle}\label{app_smoothness_eig}

We show smoothness first on $\R^d$ and then on a Riemannian manifold $(\mathcal{N}, h)$. The proof of the latter utilizes the former result in a chart, as well as smoothness results for the involved maps, which are only non-trivial in the manifold case.

\propSmoothPSrd*
\begin{proof}
    Let $p \in \R^d \setminus \mathcal{S}_{\alpha,k}$ be arbitrary. We will show that there exists a local frame of smooth vector fields spanning the subspace $\mathcal{E}^{\alpha,k}_{p'}$ at every point $p'$ on an open set $\mathscr{U}$ around $p$. By Lemma 10.32 in \cite{lee2013smooth}, this is equivalent to the subbundle being smooth on $\R^d\setminus \mathcal{S}_{\alpha,k}$.
    
    The eigenvalues of $\Sigma_{\alpha}(p)$ at p are $$\lambda_1(p) \geq \dots \geq \lambda_k(p)
   > \lambda_{k+1}(p)
   \geq 
   \dots \geq \lambda_d(p),$$ 
   where only $\lambda_k$ and $\lambda_{k+1}$ are assumed to be different. Since $\Sigma_{\alpha} : \R^d \to \R^{d\times d}$ is a smooth map, Theorem 3.1 of \cite{sun1990multiple} implies that there exists an open set $\mathscr{B}(p) \subset \R^d$ around $p$ and \textit{d} continuous functions $\bar{\lambda}_i(\cdot) : \mathscr{B}(p) \to \R$ satisfying that $\bar{\lambda}_i(p')$ is an eigenvalue of $\Sigma_{\alpha}(p')$ for all $p' \in \mathscr{B}$ and $\bar{\lambda}_i(p) = \lambda_i(p), i = 1\dots d$. 
    
    Since each $\bar{\lambda}_i$ is continuous, there exists an open subset $\mathscr{U} \subset \mathscr{B}$ on which the ordering $\bar{\lambda}_{1}(p') \geq \dots \geq \bar{\lambda}_{d}(p')$ holds for all $p' \in \mathscr{U}$, and where $\bar{\lambda}_i(p') = \bar{\lambda}_j(p')$ is only possible for $i,j$ s.t. $\bar{\lambda}_i(p) = \bar{\lambda}_j(p)$. In particular $\bar{\lambda}_i(p') < \bar{\lambda}_{k+1}(p')$ for all $i < k + 1$ and $p' \in \mathscr{U}$.
    
    Theorem 3.2 of \cite{sun1990multiple}  now says that there exists a frame of analytic vector fields $p \mapsto \left\{X_1(p), \dots, X_k(p)\right\}$ such that, for all $p'\in \mathscr{U}$, 
    \begin{align*}
    \text{span}\left\{X_1(p'), \dots, X_k(p')\right\} &= V_{\bar{\lambda}_1(p'), \dots, \bar{\lambda}_k(p')}(\Sigma_{\alpha}(p'))
    \end{align*}
    
\noindent where $V_{\bar{\lambda}_1(p'), \dots, \bar{\lambda}_k(p')}(\Sigma_{\alpha}(p'))$ denotes the eigenspace of $\Sigma_{\alpha}(p')$ corresponding to eigenvalues $\bar{\lambda}_1(p'), \dots, \bar{\lambda}_k(p')$, which is exactly the principal subbundle subspace $\mathcal{E}_{p'}^{\alpha, k}$.    
\end{proof}

To show that the principal subbundle on a Riemannian manifold is smooth, we need a result on smoothness of a certain map involving parallel transport.

\begin{lemma}\label{lem_smooth_par_transp_field}
    Let the map $f : \mathcal{N} \to \mathcal{N}$ and the vector field $O$ on $\mathcal{N}$ be smooth. Let $\Pi_{x}^y : T_x \mathcal{N} \to T_y \mathcal{N}$ denote parallel transport along the (assumed unique) length-minimizing geodesic from $x$ to $y$. Then the vector field
    \begin{align}\label{eq_smooth_parallel_vector_field}
        p\mapsto \Pi_{f(p)}^{p}O(p) \in T_p \mathcal{N}
    \end{align}
    \noindent is smooth for every $p \notin \text{Cut}(f(p))$.
\end{lemma}
\begin{proof} For $x,y \in \mathcal{N}$, the parallel transported vector $\Pi_x^y W \in T_y \mathcal{N}$ of $W \in T_x \mathcal{N}$ along a curve $\gamma : (0,1) \to \mathcal{N}$ is the value at time $1$ of a vector field $V$ along $\gamma$ satisfying the linear initial value problem (an ODE)

\begin{align}\label{eq_parallel_transport_single}
    \dot{V}^k(t) &= -V^j(t)\dot{\gamma}^i(t)\Gamma^k_{ij}(\gamma(t)) \\
    V(0) &= W,
\end{align}

\noindent where $\Gamma^k_{ij}$, $i,j,k \in \{1,\dots,d\}$ are the Christoffel symbols determined by the metric $h$. See \cite{lee2018introduction}, Section $4$, for details.

If $\gamma$ is a geodesic with initial velocity $Q \in T_x \mathcal{N}$ then it is a solution to the geodesic equations (equations (\ref{eq_parallel_transport_full_1}) and (\ref{eq_parallel_transport_full_2}), below). In this case, we can write the parallel transport equation and the geodesic equations as a single, coupled, ODE:

\begin{align}\label{eq_parallel_transport_full}
    \dot{V}^k(t) &= -V^j(t)\dot{\gamma}^i(t)\Gamma^k_{ij}(\gamma(t)) \\ \label{eq_parallel_transport_full_1}
    \dot{\gamma}^k(t) &= U^k(t) \\ \label{eq_parallel_transport_full_2}
    \dot{U}^k(t) &= -U^i(t)U^j(t)\Gamma^k_{ij}(\gamma(t))\\
    U(0) &= Q \\
    V(0) &= W \\
    \gamma(0) &= x. \label{eq_parallel_transport_full_last}
\end{align}

\noindent Note that the equation for $V$ is coupled with the equations for $\gamma$ and $U$, but not vice versa, so that, in practice, the whole path $\gamma$ can be computed first, and then subsequently $V$.

This is again a linear initial value problem, and the fundamental theorem for ODE's states that solutions exist, and depend smoothly on the initial conditions $Q,W,x$. This shows smoothness of the parallel transport operator in the case where $\gamma((0,1))$ is contained in a single chart. For the more general case, we refer to the technique used in the proof of Proposition 4.32 in \cite{lee2018introduction} for showing that solutions found on individual charts overlap smoothly. 

The map (\ref{eq_smooth_parallel_vector_field}) takes a point $p \in \mathcal{N}$ to a vector field at time $1$ satisfying equations (\ref{eq_parallel_transport_full})-(\ref{eq_parallel_transport_full_last}). For each $p$, the initial conditions are 

\begin{align*}
    x &= f(p) \\
    Q &= log^h_{f(p)}(p)\\
    W &= O(p)
\end{align*}

\noindent all of which depend smoothly on $p$, if $p \notin \text{Cut}(f(p))$. Since the solution to the ODE depends smoothly on the initial conditions, and since the initial conditions depends smoothly on $p$, the vector field (\ref{eq_smooth_parallel_vector_field}) is smooth.
\end{proof}

\propSmoothPsManif*
\begin{proof} As in the Euclidean case, we want to prove the existence of a smooth frame around every point $p \in \mathcal{S}^{'}_{\alpha, k}$ spanning the subbundle locally around $p$. We will make use of the corresponding result for $\mathcal{N} = \R^d$, in a chart. In order to do this, we need to make sure that all of the involved maps are smooth as a function of $p$. 

The tangent mean map $m : \mathcal{N} \to \mathcal{N}$ and the tensor field $p \mapsto \Sigma_{\alpha}(p) \in T_p \mathcal{N} \otimes T_p \mathcal{N}$ is smooth if each logarithm $\log^h_p(x_i)$, $i=1\dots N$, is smooth as a function of the base point $p \in \mathcal{N}$. This is ensured by the cut locus conditions in $\mathcal{S}^{'}_{\alpha,k}$.%[,,, prove it?]

Assuming smoothness of $\Sigma_{\alpha}$, we now consider charts $(U, \varphi)$ on $\mathcal{N}$ and $(O, \phi)$ on $T\mathcal{N}\otimes T\mathcal{N}$, $U\subset \R^d, \varphi : U \to \varphi(U) \subset \mathcal{N}$, respectively $O\subset \R^{d\times d}, \phi : O \to \varphi(O) \subset T\mathcal{N}\otimes T\mathcal{N}$ (identifying each $T_p\mathcal{N}\otimes T_p\mathcal{N}$ with the space of endomorphisms on $T_p\mathcal{N}$, cf. Section \ref{app_second_moment_as_tensor}), around a point $p\in \mathcal{N}$ and $\varphi(p) \in T\mathcal{N}\otimes T\mathcal{N}$. In this chart, $$f := \phi^{-1} \circ \Sigma_{\alpha, k} \circ m \circ \varphi$$ is a smooth map from $\R^d$ to $\R^{d \times d}$. Eigendecomposition of the matrix $f(p')$, $p' \in U$, is independent of the basis and thus of the choice of charts. As shown in the proof of Proposition \ref{prop_smooth_princ_subb}, there exists a smooth frame $p' \mapsto \{X_1(p'), \dots, X_k(p')\}$, $X_i(p')\in \R^{d}$, defined on some open subset $\mathscr{U} \subset \R^d$ around $\varphi^{-1}(p)$ s.t. $$\text{span}\{X_1(p'), \dots, X_k(p')\} = V_{k}(f(p')),\quad \forall p' \in \mathscr{U},$$ where the right hand side is the eigenspace of $f(p')$ corresponding to the largest $k$ eigenvalues. We have thus shown the existence of a smooth frame on $\varphi(U) \subset \mathcal{N}$ spanning the corresponding eigenspaces of $\Sigma_{\alpha} \circ m$ at every point of $\varphi(U)$. 

The last thing we need to take account of is the parallel transport map. Since parallel transport is an isometry, it holds that $$\text{span}\{\Pi_{p'}^{y} X_1(p'), \dots, \Pi_{p'}^{y} X_k(p')\} = \text{span}\{\Pi_{p'}^{y} F_1(p'), \dots, \Pi_{p'}^{y} F_k(p')\} \subset T_y\mathcal{N},$$ where $\{F\}_{i=1..k}$ is any other frame spanning the same subspace as $\{X\}_{i=1..k}$ at $p'$. Thus, the parallel transported frame $X$ spans the same subspace as the parallel transported eigenvectors $\{e_i\}_{i=1\dots k}$ at $p'$ (the $X_i$'s are not necessarily eigenvectors, as explained in \cite{sun1990multiple}). By Lemma \ref{lem_smooth_par_transp_field}, the map $p \mapsto \Pi_{m(p)}^p V(p)$ is smooth, for a smooth vector field $V$. We have thus shown that the principal subbundle at $p$ is spanned by a smooth frame around $p$.
\end{proof}

\subsection{Proof of the sub-Riemannian exponential being a local diffeomorphism on the dual subbundle}\label{app_proof_exp_local_diffeo}

We prove the result for a sub-Riemannian structure on a manifold $\mathcal{N}$. The reader may substitute $\mathcal{N} = \R^d$ if they wish.

\begin{proposition}[The exponential is a local diffeomorphism on the dual subbundle] Let $p\in \mathcal{N}$ be arbitrary. 
%\todo{Here: include conditions (James: completenessconditions,bracketconditions,smo othnessconditionsonthedistancefunction)}
There exists an open subset $C_p \subset \mathcal{D}^{\star}$ containing $0$ such that $\exp_p^{\mathcal{D}} \vert_{C_p}$ is a diffeomorphism onto its image. That is, $$M_p^{\mathcal{D}}
\defeq \exp_p^{\mathcal{D}}(C_p) \subset \mathcal{N}$$ is a smooth
$k$-dimensional embedded submanifold of $\mathbb{R}^d$ containing \textit{p}.
\end{proposition}
\begin{proof} We will show that $\exp_p^{\mathcal{D}}$ is a local immersion by showing that $d_0 \exp_p^{\mathcal{D}}$ is injective (\cite{lee2013smooth}, Proposition 4.1). For any $\eta \in T_0 \mathcal{D} \cong \mathcal{D}$ it holds that

\begin{align*}
    d_0 \left(\exp_p^{\mathcal{D}} \right) \circ \eta &=    \left.\frac{d}{ds}\right|_{s=0} \exp_p^{\mathcal{D}}(0 + s \eta) \\
    &=   \left.\frac{d}{ds}\right|_{s=0} \gamma_p^{\eta}(s) \\
    &= g^\star(p) \eta,
\end{align*}

\noindent where the second equality uses the fact that the sub-Riemannian exponential satisfies 
$\exp_p^{\mathcal{D}}(s \eta) = \gamma_p^{\eta}(s)$, see corollary 8.36 in 
\cite{agrachev2019comprehensive}. Viewed as a map $g_p^\star: \mathcal{D}^\star \to \mathcal{D}_p \subset \mathcal{N}$ (i.e. as the sub-Riemannian sharp map), $g_p^\star$ is injective on $\mathcal{D}_p^{\star}$ by construction of $\mathcal{D}_p^{\star}$. Thus $\exp_p^{\mathcal{D}}$ is an immersion. This implies the existence of a set $C_p \subset \mathcal{D}_p^{\star}$ containing 0 s.t. 
$\left.\exp_p^{\mathcal{D}}\right|_{C_p}$ is an embedding (\cite{lee2013smooth} Proposition 4.25). Which implies that $M_p^{\mathcal{D}} \defeq \exp_p^{\mathcal{D}}(C_p)$ is an embedded $k$-dimensional submanifold of $\mathcal{N}$.
$p \in M_p^{\mathcal{D}}$ since $\exp_p^{\mathcal{D}}(0) = p$, by definition. 
\end{proof}

\subsection{Expressing the second moment in coordinates}\label{app_second_moment_as_tensor}

For some $v,u \in T_p \mathcal{N}$, the expression $v \otimes u$ can be identified with an endomorphism on $T_p \mathcal{N}$. Its coordinate representation is thus a $d \times d$ matrix. There seems to be some confusion about this in the geometric statisics literature, so we give details below. We first repeat Lemma

\lemTensorCoordinates*
\begin{proof}
The tensor $v \otimes u$ is an element of the tensor product space $T_p \mathcal{N} \otimes T_p \mathcal{N}$. After choosing a Riemannian metric, there is a canonical isomorphism between $T_p \mathcal{N}$ and its dual space, $T^{\star}_p \mathcal{N}$, given by the Riemannian flat map,
$$\musFlat{} : T_p \mathcal{N} \to T^{\star}_p \mathcal{N} : u \mapsto h_p(u,\cdot) := u^{\musFlat{}}.$$ Thus $$T_p \mathcal{N} \otimes T_p \mathcal{N} \cong T_p \mathcal{N} \otimes T^{\star}_p \mathcal{N},$$ where elements of the latter space are denoted $(1,1)$ tensors. Furthermore, there is a canonical isomorphism, independent of a Riemannian metric, 
$$T_p \mathcal{N} \otimes T^{\star}_p \mathcal{N} \cong \text{End}(T_p \mathcal{N}),$$ where $\text{End}(T_p \mathcal{N})$ is the space of endomorphisms on $T_p \mathcal{N}$. This isomorphism is given by the map $\Phi$ which takes an endomorphism $A$ to the $(1,1)$ tensor $\Phi(A)$ that acts on $w \in T_p \mathcal{N}$ and $\eta \in T^{\star}_p \mathcal{N}$ by $\Phi(A)(w,\eta) = \eta(Aw)$. The linear map corresponding to a (1,1) tensor of the form $v \otimes u^{\star}, v\in T_p \mathcal{N}, u^{\star} \in T^{\star}_p\mathcal{N},$ is $w \mapsto \Phi^{-1}(v \otimes u^{\star})(w) 
= v \cdot \hspace{0.5mm} u^{\star}({w}) $, i.e. a scaling of $v$ by $u^{\star}(w) \in \R$. 

After choosing a basis for $T_p \mathcal{N}$, the tangent vectors $v, w$ can be represented as column vectors $v, w \in \R^{d \times 1}$. The flat map can be represented by the matrix $h_p$, which is the matrix representation of the Riemannian metric at \textit{p}. After identifying covectors with row vectors (i.e. coordinate representations of linear maps from $T_p \mathcal{N}$ to $\R$), $u^{\musFlat{}}$ can be represented as the row vector $u^{\musFlat{}} = (h_p u)^T \in \R^{1 \times d}$. This acts on $w$ by $u^{\musFlat{}}(w) = \left(h_p \hspace{0.5mm} u\right)^T w$. Thus, w.r.t. some chosen basis, the matrix representation of our desired endomorphism is given by 
\begin{align*}
\Phi^{-1}(v \otimes u^{\musFlat{}})= v (h_p u)^T = vu^T h_p.    
\end{align*}

\end{proof}

\subsubsection{Verifying independence of the coordinate system}

Let $Q$ be the change-of-basis matrix from basis $a$ of $T_p \mathcal{N}$ to basis $b$. Then $Q^{\star} = (Q^T)^{-1}$ is the corresponding change-of-basis matrix from basis $a^{\star}$ to $b^{\star}$ for $T^{\star}_p \mathcal{N}$, where these bases are dual to $a, b$. Thus, the change of basis of tangent vector $v$ from \textit{a} to \textit{b} is computed as $v_b = Q_{ab} v_a$. The flat map $\musFlat{}$ is a linear map from $T_p \mathcal{N}$ to $T^{\star}_p \mathcal{N}$, so if $(h_p)_a$ is its representation w.r.t. bases \textit{a} and $a^\star$, then its representation w.r.t. bases \textit{b} and $b^\star$ is computed as $$(h_p)_b = Q^{\star}(h_p)_a Q^{-1} = (Q^T)^{-1}(h_p)_a Q^{-1}.$$ 

We verify that the change of basis of the individual elements $u,v,h_p$ matches the change of basis of 
the matrix (as a linear map) (\ref{eq_matrix_repr}):
\begin{align}
    v_b u_b^T (h_p)_b &= Q v_a (Q u_a)^T (Q^T)^{-1}(h_p)_a Q^{-1} \\
    &= Q v_a u_a^T (h_p)_a Q^{-1}.
\end{align}

As opposed to this, the expression $v_b u_b^T$ does \textit{not} transform properly under basis change: $v_b u_b^T = Q v_a (Q u_a)^T  = Q v_a  u_a^T Q^T$ is only equal to $Q v_a  u_a^T Q^{-1}$ if $Q^T = Q^{-1}$, i.e. if the basis change matrix is orthogonal, meaning that it only rotates the basis.

\section{Notes on implementation}\label{app_implementation}

%[,,, her: also describe my more complicated sigmoid-kernel?]

At each step of the integration of a geodesic, eigenvectors needs to be computed at the current position \textit{p}. This involves evaluating the kernel $K_{\alpha}(\vert x_i - p \vert)$ for all $i = 1..N$. For large datasets, we suggest to do this using libraries specialized at such kernel-operations, such as KEOPS, as well as automatically filtering out points far away from $p$ whose weight will be close 0 anyway. We have not had the need to implement these optimizations in order to run the examples of Section \ref{sect_applications}.

The integration of the \textit{L} geodesics in the algorithm for the principal submanifold can be parallelized; the computation of each one is independent from the rest. Again, we have not had the need to do this for running our experiments.

\subsection{Choice of integration scheme}\label{app_integration}

The integration of Hamilton's equations can be done using a symplectic integration scheme which aims at keeping the Hamiltonian constant. A constant hamiltonian is equivalent to constant speed, cf. eq. (\ref{eq_constant_speed_ham}). This is desired because the computation of curve length and distance via eq. (\ref{eq_ham_vs_distance}) assumes constant speed. We compared ordinary Euler integration to semi-implicit Euler (see e.g. \cite{hairer2006geometric}), a first-order symplectic integrator, and found the Hamiltonian to be better preserved using ordinary Euler integration in our experiments.

\section{Choosing the kernel range $\alpha$ and bundle rank $k$}\label{app_hyperparam}

Firstly, note that these parameters can be considered to be a modelling choice, expressing the scale at which we want to analyze the data - what scale of variation to take into account. However, one can aim to find the 'lowest level of variation that is not due to random noise'. Secondly, note that the 'optimal' value of one hyperparameter depends on the value of the other. Since the rank \textit{k} takes a finite number of values $k \in \{1, \dots, d-1\}$, we suggest to start by estimating this. See \cite{bac2021scikit} for a survey and benchmarking of different methods. Given an estimated \textit{k}, we suggest to select a range parameter for which the separation between eigenvalues $\lambda_k$ and $\lambda_{k+1}$ is the most clear on average. The optimal kernel range depends on the level of noise and the rate of change of the affine subspace $\calE_p$ as a function of $p$, which, in the case of the manifold hypothesis, is an expression of the curvature of the underlying manifold. A fast varying $\calE$ calls for a smaller $\alpha$, while high levels of noise as well as a lower number of observations calls for a larger $\alpha$.

\section{Algorithm for combining principal submanifolds for 2D surface reconstruction}\label{app_combining_PS}

In this section, we present an algorithm for combining principal submanifolds $\{M^k_{\mu_j}(r_j)\}_{j=1..l}$ based at different base points $\mu_j, j=1\dots l$. In this case, $k=2$ and we'll write $M_{\mu_j}$ instead of $M^2_{\mu_j}$. Given a point $x\in \R^3$, the algorithm first projects $x$ to a set of nearest principal submanifolds, and then represents $x$ as a weighted average of these projections, weighted by the SR distance between a projection and its corresponding base point. The point $x$ can e.g. be an observation, $x \in \{x_i\}_{i=1..N}$, or a point in a principal submanifold, $x\in M_{\mu_j}$. The algorithm can then be run for each point $x$ in $\{x_i\}_{i=1..N}$ or in $M_{\mu_j}, j = 1..l$.

The point sets representing principal submanifolds $M_{\mu_j}(r_j), j=1\dots l,$ are generated by Algorithm $1$. For each point $p \in M_{\mu_j}(r_j)$, we assume that the corresponding initial cotangent $\eta(p) \in \calE_{\mu_j}^\star$ has been stored.

Apart from the hyperparameters of the principal subbundle and submanifolds, the algorithm needs a 'threshold parameter' $\epsilon > 0$. $x$ will not be projected to principal submanifold $M_{\mu_j}$ if the distance between $x$ and its projection $\hat{x}_j$ to $M_{\mu_j}$ is greater than $\epsilon$. Thus, the size of $\epsilon$ should be comparable to an estimate of the noise-level in the point cloud.

The algorithm is the following.

\begin{enumerate}
	 \item \textit{Project to each submanifold:} project $x$ to each $M_{\mu_j}(r_j), j =1..l$, wrt. Euclidean distance, i.e. find the closest point in $M_{\mu_j}(r_j)$ w.r.t. Euclidean distance. Denote this projection of $x$ to $M_{\mu_j}(r_j)$ by $\hat{x}_{j}$. Denote the corresponding initial cotangent by $\eta(\hat{x}_{j})$ and the distance by $d_{j} \defeq d(\mu_j, \eta(\hat{x}_{j})) = \Vert \eta(\hat{x}_{j}) \Vert $. 
    \item \textit{Filter out projections: } let $B \defeq \{j \in \{1, \dots, l\} \hspace{0.5mm} \big\vert \hspace{0.5mm} |x - \hat{x}_{j}|  <  \epsilon\}$ consist of indices of the basepoints satisfying that the projection of $x$ to $M_{\mu_j}$ is sufficiently close to $x$. 
    \item \textit{Rescale distances: } set $\tilde{d}_{j} \defeq d_{j} \cdot 1/s_j(d_{j})$, where $s_j$ is a continuous, decaying bijection with domain and image given by $s_j : [0, r_j] \to [0,1].$ We suggest to use the affine function satisfying these constraints.
	
    \item \textit{Compute weighted average: } the weighted representation of $x$ is now computed as $$\hat{x}_{} = \frac{1}{\sum_{j \in B} w_j}  \sum_{j \in B} w_j \hat{x}_{j},$$ where (unnormalized) weights $w_j$ are  given by $$w_j(x) = e^{-(\tilde{d}_{j} - \tilde{d}_{j^{\star}})^2/(2 \sigma)}, j=1\dots |B|,$$ and  $j^{\star} \defeq \argmin_{j \in B} d_{j}$ is the index of the  principal submanifold that is closest w.r.t. SR distance. The standard deviation $\sigma$ in $w_j$ controls how fast the weights should go to zero. A general-purpose choice is $\sigma = \max_{j \in \{1, .., l\}} r_{j}\}$.
	  
    \end{enumerate}

\section{Supplementary figures}\label{app_PS_experiments}

\subsection{Illustration of overlapping submanifolds}\label{app_overlapping}

Figure \ref{fig_overlap} is a supplement to figure \ref{fig_surface_recon}, zooming in on the region of overlap between the two principal submanifolds. 

\begin{figure}[h]
  \centering
  \includegraphics[scale=0.35]{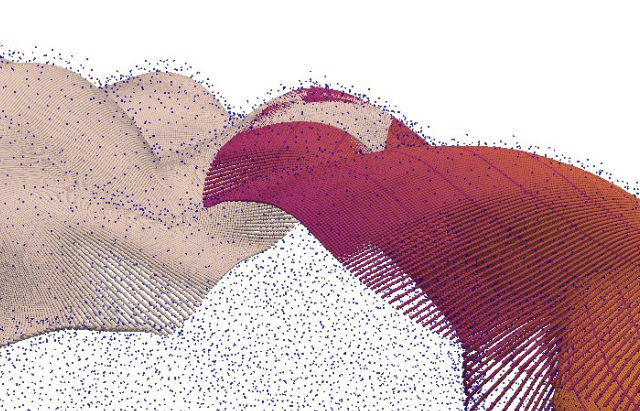}
  \caption{Supplementary figure to Figure \ref{fig_surface_recon}, zooming in on the region where the two submanifolds overlap. The left, beige submanifold in this figure is the purple one in Figure \ref{fig_surface_recon}, recolored to be able to distinguish more clearly the two submanifolds.}
  \label{fig_overlap}
\end{figure}

\subsection{Reconstruction of head sculpture surface under noise level 3 out of 3}\label{app_surface_recon_high_noise}

Figure \ref{fig_surface_recon_high_noise} illustrates the reconstruction of the face of the 'head sculpture' (from the benchmark dataset described in \cite{surveySurfacReconstructionHuang}), with noise level 3 out of 3. The parameters are the same as for the experiment described in section \ref{sect_surface_recon} except for a slightly larger kernel range.

\begin{figure}[h!]
  \centering
  \begin{tabular}{@{}c@{}}
    \includegraphics[scale=0.33]{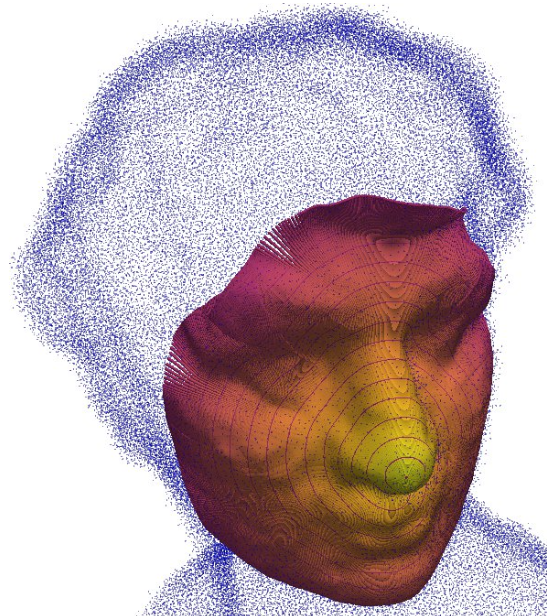} \\[\abovecaptionskip]
    \small (a) Frontal view.
  \end{tabular}

  \vspace{\floatsep}

  \begin{tabular}{@{}c@{}}
    \includegraphics[scale=0.33]{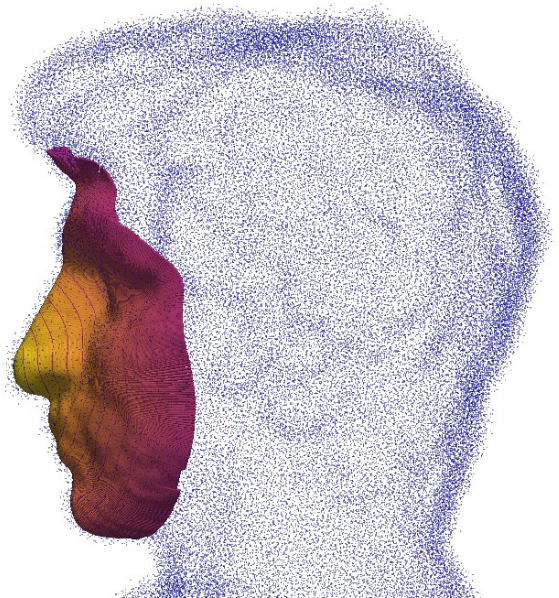} \\[\abovecaptionskip]
    \small (b) Side view.
  \end{tabular}

  \caption{Figures \textit{(a)} and \textit{(b)} show a principal submanifold recontructing the 'head sculpture' surface from a point cloud (purple points) with noise level 3 out of 3. The submanifold is centered around the tip of the nose. The figure shows the raw points generated by Algorithm \ref{algo_principal_submanifold} - no subsequent processing, apart from coloring, has been applied. The skewed circles on the face are geodesic balls, i.e. points on the same circle has the same SR distance to the center point. The colors of the face depends on the SR distance to the base point at the tip, a lighter color signifies shorter distance.}\label{fig_surface_recon_high_noise}
\end{figure}

\subsection{Illustration of the log map on a 4-dimensional sphere in $\R^{50}$}\label{app_log}

Figure \ref{fig_log_sphere} shows a single computed geodesic, found by solving the log problem $\log_p(q)$, for $p,q$ and observations as described in section\ref{sect_learning_distance_metric}. The distance $d^\calE(p,q)$ is estimated as the length of the computed geodesic. The blue points are observations on the $4$-dimensional sphere embedded in $\R^50$ projected to $\R^3$.

\begin{figure}[h!]
  \centering
  \includegraphics[scale=0.5]{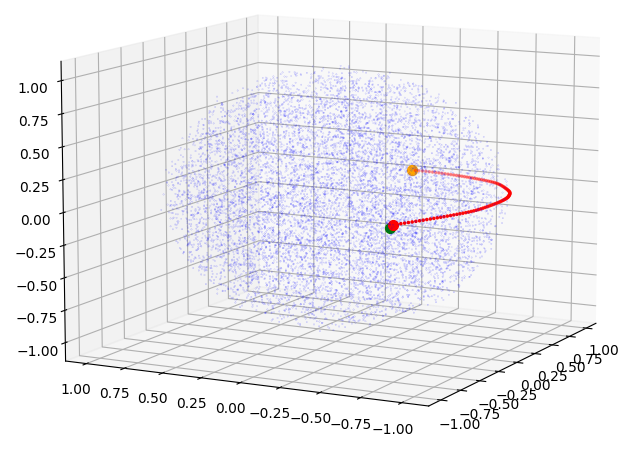}
  \caption{Illustration of a single computed geodesic found by solving the log problem $\log_p(q)$, for $p,q$ and observations as described in section \ref{sect_learning_distance_metric}. The blue points are observations on the $4$-dimensional sphere embedded in $\R^{50}$ projected to $\R^3$. The orange point is the initial point $p$, the red points are points on the geodesic, the green point is the target point $q$.}
  \label{fig_log_sphere}
\end{figure}

\end{document}